\documentclass[11pt]{article}
\pdfoutput=1 
\usepackage[sc]{mathpazo}

\usepackage{palatino,microtype}
\usepackage[margin=15pt,font=small,labelfont={bf}]{caption}
\usepackage[margin=1in]{geometry}
\usepackage[english]{babel}
\usepackage[utf8]{inputenc}
\usepackage[compact]{titlesec}
\usepackage{cmap}
\usepackage[T1]{fontenc}
\usepackage{bm}
\usepackage{booktabs}
\usepackage{mathtools}
\usepackage{amsfonts}
\usepackage{amsmath}
\usepackage{amssymb}
\usepackage{amsthm}
\usepackage{float}
\usepackage{graphics}
\usepackage{natbib}
\usepackage[backref=page]{hyperref}
\usepackage[svgnames]{xcolor}
\hypersetup{colorlinks={true},urlcolor={blue},linkcolor={DarkBlue},citecolor=[named]{DarkGreen},linktoc=all}
\renewcommand*{\backref}[1]{}
\renewcommand*{\backrefalt}[4]{%
    \ifcase #1 (Not cited.)%
    \or        (Cited on page~#2)%
    \else      (Cited on pages~#2)%
    \fi}

\usepackage{microtype}
\usepackage[capitalise,nameinlink,noabbrev]{cleveref}
\usepackage{doi}

\newcommand{\BibTeX}{\rm B\kern-.05em{\sc i\kern-.025em b}\kern-.08em\TeX}
\usepackage[justification = centering]{subcaption}
\usepackage{tabularx}
\usepackage{thmtools,thm-restate}
\usepackage[linesnumbered, ruled,vlined]{algorithm2e}
\SetKwComment{Comment}{/* }{ */}

\newtheorem{observation}{Observation}
\newtheorem{remark}{Remark}
\newtheorem{corollary}{Corollary}
\newtheorem{proposition}{Proposition}

\usepackage[capitalise,nameinlink,noabbrev]{cleveref}
\crefname{algocf}{alg.}{algs.}
\Crefname{algocf}{Algorithm}{Algorithms}
\Crefname{claim}{Claim}{Claims}
\Crefname{corollary}{Corollary}{Corollaries}
\Crefname{definition}{Definition}{Definitions}
\Crefname{example}{Example}{Examples}
\Crefname{lemma}{Lemma}{Lemmas}
\Crefname{property}{Property}{Properties}
\Crefname{proposition}{Proposition}{Propositions}
\Crefname{remark}{Remark}{Remarks}
\Crefname{theorem}{Theorem}{Theorems}

\newcommand{\ConstructiveExistsAddMen}{\textup{\textsc{Add Men To Match Pair}}}
\newcommand{\ConstructiveExistsDelete}{\textup{\textsc{Delete to Match Pair}}}
\newcommand{\ConstructiveExistsDeleteMen}{\textup{\textsc{Delete Men To Match Pair}}}
\newcommand{\MultipleConstructiveExistsDeleteMen}{\textup{\textsc{Delete Men To Match one of Multiple Pairs}}}

\newcommand{\UnrestrictedAddCapacityMatchPair}{\textup{\textsc{Add Capacity To Match Pair}}}
\newcommand{\RestrictedAddCapacityMatchPair}{\textup{\textsc{Budgeted Add Capacity To Match Pair}}}

\newcommand{\UnrestrictedDeleteCapacityMatchPair}{\textup{\textsc{Delete Capacity To Match Pair}}}
\newcommand{\RestrictedDeleteCapacityMatchPair}{\textup{\textsc{Budgeted Delete Capacity To Match Pair}}}

\newcommand{\ExactExistsAddMen}{\textup{\textsc{Add Men To Stabilize Matching}}}
\newcommand{\ExactExistsAdd}{\textup{\textsc{Exact-Exists-Add}}}
\newcommand{\ExactExistsDelete}{\textup{\textsc{Exact-Exists-Delete}}}
\newcommand{\ExactExistsDeleteMen}{\textup{\textsc{Delete Men To Stabilize Matching}}}

\newcommand{\RestrictedAddCapacityStabilizeMatching}{\textup{\textsc{Budgeted Add Capacity To Stabilize Matching}}}
\newcommand{\UnrestrictedAddCapacityStabilizeMatching}{\textup{\textsc{Add Capacity To Stabilize Matching}}}

\newcommand{\RestrictedDeleteCapacityStabilizeMatching}{\textup{\textsc{Budgeted Delete Capacity To Stabilize Matching}}}
\newcommand{\UnrestrictedDeleteCapacityStabilizeMatching}{\textup{\textsc{Delete Capacity To Stabilize Matching}}}

\newcommand{\I}{\mathcal{I}}
\newcommand{\N}{\mathbb{N}}

\newcommand{\Polytime}{\textrm{\textup{Poly time}}}

\newcommand{\NPH}{\textrm{\textup{NP-hard}}}

\newcommand{\FPDA}{{\sf FPDA}}
\newcommand{\WPDA}{{\sf WPDA}}

\newcommand{\whole}{\N \cup \{0\}}
\renewcommand{\>}{\succ}

\colorlet{thechosenone}{red}
\colorlet{thechosentwo}{blue}
\makeatletter
\newcommand*{\RN}[1]{\expandafter\@slowromancap\romannumeral #1@}
\makeatother
\usepackage{lineno}

\newif\ifcomments
\commentstrue
\ifcomments
\newcommand{\RV}[1]{{\textcolor{red}{RV: { #1}}}}
\else
\newcommand{\RV}[1]{}
\fi

\newif\ifcomments
\commentstrue
\ifcomments
\newcommand{\SN}[1]{{\color{blue!50!magenta}{SN: }{#1} }}
\else
\newcommand{\SN}[1]{}
\fi

\newif\ifcomments
\commentstrue
\ifcomments
\newcommand{\SG}[1]{{\textcolor{green}{SG: { #1}}}}
\else
\newcommand{\SG}[1]{}
\fi

\usepackage{tikz}
\usetikzlibrary{calc,arrows,matrix,decorations.pathreplacing}
\usepackage{paralist}

\usepackage{xcolor,colortbl,booktabs}
\colorlet{myred}{red!25}
\colorlet{myblue}{blue!25}
\colorlet{mygreen}{green!25}
\colorlet{mygrey}{gray!15}

\usepackage{multirow,array}

\usepackage{framed}

\newcommand{\problembox}[3]{
\smallskip
\begin{center}
	{\small 
		\begin{tabularx}{1.0\columnwidth}{ll}
			\toprule
			\multicolumn{2}{c}{#1} \\
			\midrule
			\textbf{Given:}& \parbox[t]{0.75\columnwidth}{#2
				\vspace*{1mm}} \\%
			\textbf{Question:}& \parbox[t]{0.75\columnwidth}{#3} \\ 
			\bottomrule
		\end{tabularx}
	}
\end{center}
\smallskip
}

\usepackage{thmtools,thm-restate}
\newtheorem{example}{Example}
\newtheorem{lemma}{Lemma}
\newtheorem{theorem}{Theorem}
\Crefname{claim}{Claim}{Claims}
\Crefname{claim}{Claim}{Claims}
\Crefname{corollary}{Corollary}{Corollaries}
\Crefname{definition}{Definition}{Definitions}
\Crefname{example}{Example}{Examples}
\Crefname{lemma}{Lemma}{Lemmas}
\Crefname{property}{Property}{Properties}
\Crefname{proposition}{Proposition}{Propositions}
\Crefname{remark}{Remark}{Remarks}
\Crefname{theorem}{Theorem}{Theorems}

\title{Capacity Modification in the Stable Matching Problem\thanks{A preliminary version of this manuscript has appeared in the Proceedings of the 23rd International Conference on Autonomous Agents and Multiagent Systems~\citep{GSN+24capacity}.}}
\author{
	\begin{tabular}{m{0.25\linewidth}m{0.15\linewidth}m{0.15\linewidth}m{0.25\linewidth}}
		\multicolumn{2}{c}{\textbf{Salil Gokhale}} & \multicolumn{2}{c}{\textbf{Shivika Narang}}\\
		\multicolumn{2}{c}{\small{IIT Delhi, India}} & \multicolumn{2}{c}{\small{UNSW Sydney, Australia}}\\
		\multicolumn{2}{c}{\href{mailto:salilgokhale24@gmail.com}{\small{\texttt{Salil.Gokhale.mt121@maths.iitd.ac.in}}}} & \multicolumn{2}{c}{\href{mailto:s.narang@unsw.edu.au}{\small{\texttt{s.narang@unsw.edu.au}}}}\\
        &&&\\
        \multicolumn{2}{c}{\textbf{Samarth Singla}} & \multicolumn{2}{c}{\textbf{Rohit Vaish}}\\
		\multicolumn{2}{c}{\small{IIT Delhi, India}} & \multicolumn{2}{c}{\small{IIT Delhi, India}}\\
		\multicolumn{2}{c}{\href{mailto:mt6210942@iitd.ac.in}{\small{\texttt{mt6210942@iitd.ac.in}}}} & \multicolumn{2}{c}{\href{mailto:rvaish@iitd.ac.in}{\small{\texttt{rvaish@iitd.ac.in}}}}\\
	\end{tabular}
}

\date{}

\begin{document}

\maketitle 


\begin{abstract}
\noindent We study the problem of capacity modification in the many-to-one stable matching of workers and firms. Our goal is to systematically study how the set of stable matchings changes when some seats are added to or removed from the firms. We make three main contributions: First, we examine whether firms and workers can improve or worsen upon changing the capacities under worker-proposing and firm-proposing deferred acceptance algorithms. Second, we study the computational problem of adding or removing seats to match a fixed worker-firm pair in some stable matching or make a fixed matching stable for the modified problem. We develop polynomial-time algorithms for these problems when only the overall change in the firms' capacities is restricted and show NP-hardness when additional constraints exist for individual firms. Lastly, we compare capacity modification with the classical model of preference manipulation by firms and identify scenarios under which one mode of manipulation outperforms the other. We find that a threshold associated with a firm's capacity, which we call its \emph{peak}, crucially determines the effectiveness of different manipulation actions.
\end{abstract}

\section{Introduction}
\label{sec:Introduction}

The stable matching problem is a classical problem at the intersection of economics, operations research, and computer science~\citep{GS62college,GI89stable,RS90two,K97stable,M13algorithmics}. The problem involves two sets of agents: workers and firms, each with a preference ordering over the agents on the other side. The goal is to find a matching that is \emph{stable}, i.e., one where no worker-firm pair prefer each other over their current matches. In recent years, the stable matching problem has gained significant attention in the artificial intelligence literature. Numerous studies have focused on developing and analyzing new models of the problem~\citep{KHI+17controlled,VG17manipulating,HUV21accomplice,SDT21coalitional,BBH+21bribery,BHS24map}. Additionally, researchers have explored the application of computational tools, including SAT/SMT solvers and ILP solvers, to tackle computationally intractable versions of the problem~\citep{drummond2015sat,pettersson2021improving,cseh2022collection}.

Many real-world matching markets have been influenced by the stable matching problem, such as school choice~\citep{AS03school,APR05new,APR+05boston}, entry-level labor markets~\citep{R84evolution,RP99redesign}, and refugee resettlement~\citep{AAM+21placement,ACG+18stability}. In these applications, each agent on one side of the market (e.g., the firms) has a \emph{capacity constraint} that limits the maximum number of agents on the other side (namely, workers) with which it can feasibly match. Remarkably, for any given capacities, a stable matching of workers and firms always exists and can be computed using the celebrated deferred-acceptance algorithm~\citep{GS62college,R84evolution}. 

Although the stable matching problem assumes fixed capacities, it is common to have \emph{flexible} capacities in practice. This feature is particularly useful in scenarios with fluctuating demand or popularity, such as vaccine distribution or course allocation. For example, in 2016, nineteen colleges at Delhi University in India increased their total capacity by 2000 seats across various courses~\citep{HT16}. Another example is the \emph{ScheduleScout} platform,\footnote{\url{https://www.getschedulescout.com/}} formerly known as Course Match~\citep{BCK+17course}, used in course allocation at the Wharton School. This platform allows adding or removing seats in courses that are either undersubscribed or oversubscribed, respectively.\footnote{\url{https://www.youtube.com/watch?v=OSOanbdV3jI&t=1m38s}} Flexible capacities also allow for accommodating other goals such as Pareto optimality or social welfare~\citep{TM22quota}. In more complex matching environments such as stable matching with couples where a stable solution is not guaranteed to exist, a small change in the capacities can provably restore the existence of a stable outcome~\citep{NV18near}. Thus, flexible capacities are well-motivated from both theoretical and practical perspectives. We will use the term \emph{capacity modification} to refer to the change in the capacities of the firms by a central planner. 

The theoretical study of capacity modification was initiated by \citet{S97manipulation}, who showed that under any stable matching algorithm, there exists a scenario where some firm is better off upon reducing its capacity. The computational aspects of capacity modification have also recently gained attention~\citep{BCL+22capacity,CC23optimal,BCL+23capacity}. However, some natural questions about how the set of stable matchings responds to changes in capacities remain unanswered. Specifically, can a given worker-firm pair get matched under some stable matching by modifying the capacities? Or, can a given matching be realized as a stable outcome of the modified instance? Furthermore, if we consider the perspective of a strategic firm, there has been a lack of a clear comparison between ``manipulation through capacity modification'' and the traditional approach of ``manipulation through misreporting preferences''. Our interest in this work is to address these gaps.

\subsection*{Our Contributions}
The idea of modifying agent capacities in order to achieve a desired objective is quite recent. Many fundamental questions in this space have been unexplored so far. In this paper, we undertake a systematic analysis of the structural and computational aspects of the capacity modification problem. Our technical contributions lie in three areas:

\paragraph{Capacity modification trends.} In \Cref{sec:Capacity_Modification_Trends}, as a warm up of sorts, we study the effect of capacity modification on workers and firms. Specifically, if a firm were to increase its capacity by exactly one, how would it and the set of workers be affected? We observe that increasing a firm's capacity by $1$ can, in some cases, improve and, in other cases, worsen its outcome under both worker-proposing and firm-proposing deferred acceptance algorithms. The workers, on the other hand, can improve but never worsen (see \Cref{tab:Summary} in \Cref{sec:Capacity_Modification_Trends}).

\paragraph{Computational results.} In \Cref{sec:Computational_Results}, we study a natural computational problem faced by a central planner: Given a many-to-one instance, how can a {\em fixed number} of seats be added to (similarly, removed from) the firms to either match a fixed worker-firm pair in some stable matching or make a given matching stable in the new instance? We show that these problems admit polynomial-time algorithms. We also study a generalization where individual firms have constraints on the seats added to or removed from them, in addition to an aggregate budget. Interestingly, in this case, the problem of matching a fixed worker-firm pair turns out to \NPH{}. By contrast, the problem of making a given matching stable can still be efficiently solved (see \Cref{tab:manipulation} in \Cref{sec:Computational_Results}).

\paragraph{Capacity modification v/s preference manipulation.} In \Cref{sec:Capacity_vs_Preference}, we examine which mode of manipulation is more useful 
for a strategic firm: underreporting/overreporting capacity or misreporting preferences. We introduce a new concept called the {\em peak} of a given firm, which is an instance-dependent threshold on its capacity. Our analysis shows that the effectiveness of each manipulation action (i.e., adding/deleting capacity or misreporting preferences) depends crucially on whether the firm's capacity is below, at, or above peak (see \Cref{fig:comparison} in \Cref{sec:Capacity_Modification_Trends}). For a firm to successfully manipulate its preferences, its capacity must be strictly below its peak~(under the worker-proposing algorithm) or at most its peak~(when firms propose). Thus, the concept of peak appears to be relevant beyond capacity modification.

\subsection*{Related Work}
The stable matching problem has inspired a large body of work in economics, operations research, computer science, and artificial intelligence~\citep{GS62college,GI89stable,RS90two,K97stable,M13algorithmics}. 

\paragraph{Strategic aspects.} Prior work in economics and operations research has demonstrated strategic vulnerabilities of stable matching algorithms. It is known that any stable matching algorithm is susceptible to manipulation via misreporting of preferences~\citep{DF81machiavelli,R82economics}, underreporting of capacities~\citep{S97manipulation}, and formation of prearranged matches~\citep{S99can}.\footnote{In prearranged matches, a worker and firm can choose to match outside the algorithm. The worker does not participate in the algorithm and, in return, is offered a seat at the firm. The firm then has one less seat available through the algorithm.}
Subsequently, \citet{RP99redesign} showed via experiments on the National Resident Matching Program data that less than $1\%$ of the programs could benefit by misreporting preferences or underreporting capacities. \citet{KP09incentives} provided theoretical justification for these findings by showing that incentives for such manipulations vanish in large markets. Note that, unlike the above results that only apply to specific datasets~\citep{RP99redesign} or in the asymptotic setting~\citep{KP09incentives}, our algorithmic results provide worst-case guarantees for \emph{any} given instance.

\paragraph{Restricted preferences.} Another line of work has explored restricted preference domains to circumvent the susceptibility of stable matching mechanisms to manipulations~\citep{KU06games,kojima2007can,K12two}. In particular, \citet{KU06games} have shown that under strongly monotone preferences (formally defined in \Cref{sec:Preliminaries}), a firm cannot manipulate by underreporting its capacity under the {\em worker-proposing algorithm}. However, manipulation via underreporting is possible under other algorithms, such as the firm-proposing algorithm.


\paragraph{Computational aspects.}
The computational problem of modifying capacities to achieve a given objective has received significant attention in recent years. 
\citet{BCL+22capacity} showed that the problem of adding (similarly, removing) seats from the firms to minimize the average rank of matched partners of the workers is \NPH{} to approximate within $\mathcal{O}(\sqrt{m})$, where $m$ is the number of workers. \citet{BCL+23capacity} developed a mixed integer linear program for this problem. 

\citet{CC23optimal} studied the problem of increasing the firms' capacities to obtain a stable and perfect matching, and similarly, a matching that is stable and Pareto efficient for the workers. They considered two objectives for this problem: minimizing the overall increase in the firms' capacities and minimizing the maximum increase in any firm's capacity. 
\citet{ADV24capacity} 
studied the problem of adding seats to firms to achieve a matching that is stable (with respect to the modified capacities) and not Pareto dominated (as per workers' preferences only) by any other stable matching. Some of our computational results draw upon the work of \citet{BBH+21bribery}, who studied the control problem for stable matchings in the one-to-one setting. We discuss the connection with this work in~\Cref{sec:Computational_Results}.


\section{Preliminaries}
\label{sec:Preliminaries}
%
For any positive integer $r$, let $[r] \coloneqq \{1,2,\dots,r\}$.
\paragraph{Problem instance.} 
An instance of the \emph{many-to-one} matching problem is given by a tuple~${\langle F,W,C, \> \rangle}$, where $F = \{f_1,\dots,f_n\}$ is the set of $n \in \mathbb{N}$ \emph{firms}, $W = \{w_1,\dots,w_m\}$ is the set of $m \in \mathbb{N}$ \emph{workers}, $C = \{c_1,\dots,c_n\}$ is the set of \emph{capacities} of the firms (where, for every $i \in [n]$, $c_i \in \N \cup \{0\}$ is the capacity of $f_i$), and $\> = (\>_{f_1},\dots,\>_{f_n}, \>_{w_1},\dots,\>_{w_m})$ is the \emph{preference profile} consisting of the ordinal preferences of all firms and workers. Each worker $w \in W$ is associated with a linear order (i.e., a strict and complete ranking) $\>_w$ over the set $F \cup \{\emptyset\}$. Each firm $f \in F$ is associated with a linear order $\>_f$ over the set $W \cup \{\emptyset\}$. We will use the term \emph{agent} to refer to a worker or a firm, i.e., an element in the set $W \cup F$.

For two capacity vectors $C, \overline{C} \in (\N \cup \{0\})^n$, we will write $\overline{C} \geq C$ to denote coordinate-wise greater than or equal to, i.e., for every $i \in [n]$, $\overline{c}_i \geq c_i$, where $c_i$ and $\overline{c}_i$ are the $i^\textup{th}$ coordinate of vectors $C$ and $\overline{C}$, respectively. Additionally, we will write $|\overline{C} - C|_{1}$ to denote the $L^1$ norm of the difference vector, i.e., $|\overline{C} - C|_{1} \coloneqq \sum_{i=1}^n |\overline{c}_i - c_i|$.

When all firms have unit capacities (i.e., for each firm $f \in F, c_f = 1$), we obtain the \emph{one-to-one} matching problem. In this case, we will follow the terminology from the literature on the stable marriage problem~\citep{GS62college} and denote a problem instance by $\langle P,Q, \> \rangle$, where $P$ and $Q$ denote the set of $n$ \emph{men} and $m$ \emph{women}, respectively, and $\>$ denotes the corresponding preference profile.

\paragraph{Complete preferences.} A worker $w$ is said to be \emph{acceptable} to a firm $f$ if $w \>_f \emptyset$. A set of workers $S \subseteq W$ is acceptable to a firm $f$, denoted by $S \>_f \emptyset$, if all workers in $S$ are acceptable to $f$. Likewise, a firm $f$ is acceptable to a worker $w$ if $f \>_w \emptyset$. An agent's preferences are said to be \emph{complete} if all agents on the other side are acceptable to it.

\paragraph{Responsive preferences.} Throughout the paper, we will assume that firms' preferences over subsets of workers are \emph{responsive}~\citep{R85college}. This means that 
for any subsets $S,S' \subseteq W$ of workers where $S$ is derived from $S'$ by replacing a worker $w' \in S'$ with a more preferred worker $w$, it must be that $S \>_f S'$. More formally, the extension of firm $f$'s preferences over subsets of workers is responsive if, for any subset $S \subseteq W$ of workers,
\begin{itemize}
    \item for all $w \in W \setminus S$, $S\cup \{w\} \>_f S$ if and only if $w \>_f \emptyset$, and
    \item for all $w,w' \in W \setminus S$, $S \cup \{w\} \>_f S \cup \{w'\}$ if and only if $w \>_f w'$.
\end{itemize}

We will write $S \succeq_f S'$ to denote that either $S \>_f S'$ or $S = S'$. Further, we will always consider the transitive closure of any responsive extension of $\>_f$, which, in turn, induces a partial order over the set of all subsets of workers.

We will now define two subdomains of responsive preferences that will be of interest to us: \emph{strongly monotone} and \emph{lexicographic}.


\paragraph{Strongly monotone preferences.}
A firm is said to have \emph{strongly monotone} preferences~\citep{KU06games} if its preferences are responsive and it prefers cardinality-wise larger subsets of workers. That is, for any pair of acceptable subsets of workers $S,T$ such that $|S| > |T|$, it holds that $S \>_f T$.

\paragraph{Lexicographic preferences.}
A firm $f$ is said to have \emph{lexicographic} preferences if it prefers any subset of workers containing its favorite worker over any subset not containing it, subject to which it prefers any subset containing its second-favorite worker over any subset not containing it, and so on. Formally, given a linear order $\>_f$ over the set $W \cup \{\emptyset\}$ and any pair of distinct acceptable subsets of workers $S$ and $T$, we have $S \>_f T$ if and only if the favorite worker of firm $f$ (as per $\>_f$) in the set difference of $S$ and $T$ (i.e., $S \setminus T \cup T \setminus S$) lies in $S$. 
Observe that lexicographic preferences are responsive.

For many-to-one instances with two workers (i.e., $|W|=2$) acceptable to a firm, lexicographic and strongly monotone preferences coincide. However, for instances with three or more workers, strongly monotone preferences are not lexicographic, and lexicographic preferences are not strongly monotone.\footnote{This can be easily seen by considering a firm with preference over singletons as $w_1 \succ w_2 \succ w_3 \succ \cdots$. A firm with lexicographic preferences will prefer $\{w_1\}$ over $\{w_2, w_3\}$. On the other hand, under strongly monotone preferences, the firm will prefer $\{w_2, w_3\}$ over $\{w_1\}$. Hence, lexicographic and strongly monotone preferences do not coincide when there are three or more workers.}


\paragraph{Many-to-one matching.} Given an instance $\I = \langle F,W,C, \> \rangle$, a many-to-one matching for $\I$ is specified by a function $\mu: F \cup W \rightarrow 2^{F \cup W}$ such that:
\begin{itemize}
    \item for every firm $f \in F$, $|\mu(f)| \leq c_f$ and $\mu(f) \subseteq W$, i.e., each firm $f$ is matched with at most $c_f$ workers,
    \item for every worker $w \in W$, $|\mu(w)| \leq 1$ and $\mu(w) \subseteq F$, i.e., each worker is matched with at most one firm, and
    \item for every worker-firm pair $(w,f) \in W \times F$, $\mu(w) = \{f\}$ if and only if $w \in \mu(f)$.
\end{itemize}

A firm $f$ with capacity $c_f$ is said to be \emph{saturated} under the matching $\mu$ if $|\mu(f)| = c_f$; otherwise, it is said to be \emph{unsaturated}.

For simplicity, we will use the term \emph{matching} instead of `many-to-one matching' whenever it is clear from context. We will explicitly use the qualifiers `one-to-one' and `many-to-one' when the distinction between the two notions is relevant to the context. 



\paragraph{Stability.} A many-to-one matching $\mu$ is said to be
\begin{itemize}
    \item \emph{blocked by a firm $f$} if there is some worker $w \in \mu(f)$ such that $\emptyset \>_f \{w\}$. That is, firm $f$ prefers to keep a seat vacant rather than offer it to worker $w$.
    \item \emph{blocked by a worker $w$} if $\emptyset \>_w \mu(w)$. That is, worker $w$ prefers being unmatched over being matched with firm $\mu(w)$.
    \item \emph{blocked by a worker-firm pair $(w,f)$} if worker $w$ prefers being matched with firm $f$ over its current outcome under $\mu$, and, simultaneously, firm $f$ prefers being matched with worker $w$ along with a subset of the workers in $\mu(f)$ over being matched with the set $\mu(f)$. That is, $f \>_w \mu(w)$ and there exists a subset $S \subseteq \mu(f)$ such that $S \cup \{w\} \>_f \mu(f)$ and $|S \cup \{w\}| \leq c_f$.\footnote{One might ask about blocking coalitions, wherein a set of workers and firms together block a given matching. It is known that if a coalition of workers and firms blocks a matching, then so does some worker-firm pair~\citep[Theorem 3.3]{RS90two}.}
    %
    %
    \item \emph{stable} if any worker, firm, or worker-firm pair does not block it.
\end{itemize}
The set of stable matchings for the instance $\I$ is denoted by $\S_{\I}$. Note that the above definition of stability assumes responsive preferences. A more general definition of stability in terms of choice sets can be found in~\citep{S97manipulation}.


\paragraph{Firm and worker-optimal stable matchings.} It is known that given any many-to-one matching instance, there always exists a \emph{firm-optimal} (respectively, \emph{worker-optimal}) stable matching that is weakly preferred by all firms (respectively, all workers) over any other stable matching. This result, due to \citet{R84evolution}, is recalled in \Cref{prop:WorkerOptimal_FirmOptimal} below. 
We will write FOSM and WOSM to denote the firm-optimal and worker-optimal stable matching, respectively.

\begin{restatable}[Firm-optimal and worker-optimal stable matching;~\citealp{R84evolution}]{proposition}{WorkerOptFirmOpt}
Given any instance $\I$, there exist (not necessarily distinct) stable matchings $\mu_F, \mu_W \in \S_\I$ such that for every stable matching $\mu \in \S_\I$, $\mu_F(f) \succeq_f \mu(f) \succeq_f \mu_W(f)$ for every firm $f \in F$ and $\mu_W(w) \succeq_w \mu(w) \succeq_w \mu_F(w)$ for every worker $w \in W$.
\label{prop:WorkerOptimal_FirmOptimal}
\end{restatable}

\paragraph{Worker-proposing and firm-proposing algorithms.} Two well-known algorithms for finding stable matchings are the worker-proposing and firm-proposing deferred acceptance algorithms, denoted by \WPDA{} and \FPDA{}, respectively. The \WPDA{} algorithm proceeds in rounds, with each round consisting of a \emph{proposal} phase followed by a \emph{rejection} phase. In the proposal phase, every unmatched worker proposes to its favorite acceptable firm that hasn't rejected it yet. Subsequently, in the rejection phase, each firm $f$ tentatively accepts its favorite $c_f$ proposals and rejects the rest. The algorithm continues until no further proposals can be made. 

Under the \FPDA{} algorithm, firms make proposals and workers do the rejections. Each firm makes (possibly) multiple proposals in each round according to its ranking over individual workers. Each worker tentatively accepts its favorite proposal and rejects the rest. \citet{R84evolution} showed that the \WPDA{} and \FPDA{} algorithms return the worker-optimal and firm-optimal stable matchings, respectively.

\paragraph{Rural hospitals theorem.} The \emph{rural hospitals theorem} is a well-known result which states that, for any fixed firm $f$, the \emph{number} of workers matched with $f$ is the same in \emph{every} stable matching~\citep{R84evolution}. Furthermore, if $f$ is unsaturated in any stable matching, then it is matched with the same \emph{set} of workers in \emph{every} stable matching~\citep{R86allocation}.

\begin{restatable}[Rural hospitals theorem;~\citealp{R84evolution},~\citealp{R86allocation}]{proposition}{RuralHosp}
Given any instance $\I$, any firm $f$, and any pair of stable matchings $\mu, \mu' \in \S_\I$, we have that $|\mu(f)| = |\mu'(f)|$. Furthermore, if $|\mu(f)| < c_f$ for some stable matching $\mu \in \S_\I$, then $\mu(f) = \mu'(f)$ for every other stable matching $\mu' \in \S_\I$.
\label{prop:RuralHospitals}
\end{restatable}

\paragraph{Canonical one-to-one instance.} Given a many-to-one instance $\I = \langle F,W,C, \> \rangle$ with responsive preferences, there exists an associated one-to-one instance $\I' = \langle P,Q,\>' \rangle$ obtained by creating $c_f$ men for each firm $f$ and one woman for each worker. Each man's preferences for the women mirror the corresponding firm's preferences for the corresponding workers. Each woman prefers all men corresponding to a more preferred firm over all men corresponding to any less preferred firm (according to the corresponding worker's preferences). For any fixed firm, all women prefer the man corresponding to its first copy over the man representing its second copy, and so on. Any stable matching in the one-to-one instance $\I'$ maps to a unique stable matching in the many-to-one instance $\I$, obtained by ``compressing'' the former matching in a natural way (see \Cref{example:canonical} in the appendix).

\begin{restatable}[Canonical instance;~\citealp{GS85some}]{proposition}{Canonical}
Given any many-to-one instance $\I = \langle F,W,C, \> \rangle$, there exists a one-to-one instance $\I' = \langle P,Q,\>' \rangle$ such that there is a bijection between the stable matchings of $\I$ and $\I'$. Furthermore, the instance $\I'$ can be constructed in polynomial time.
\label{prop:Canonical}
\end{restatable}

\section{Warmup: How Does Capacity Modification Affect Workers and Firms?}
\label{sec:Capacity_Modification_Trends}

In this section, we study how changing the capacity of a firm can affect the outcomes under stable matchings for the firm and the workers. Specifically, we consider the worker-proposing and firm-proposing algorithms (\WPDA{} and \FPDA{}) and ask if a firm can improve/worsen when a unit capacity is added to it. Similarly, we will ask whether all workers can improve or if some worker can worsen when a firm's capacity is increased. \Cref{tab:Summary} summarizes these trends.

\begin{table}[t]
\centering
\begin{tabular}{ccc}
%
& \WPDA{} & \FPDA{}\\
\cmidrule{2-3}
\rowcolor{mygrey}
 & Yes & Yes \\
\rowcolor{mygrey} \multirow{-2}{*}{Can the firm improve?} & [\Cref{eg:MasterList-AllWorkersImprove}] & [\Cref{eg:MasterList-AllWorkersImprove}] \\
%
\multirow{2}{*}{Can the firm worsen?} & Yes
& Yes \\
& [\Cref{example:LP-firm-worse}]$^{\S}$&  [\Cref{example:LP-firm-worse}]$^{\S}$\\
%
\rowcolor{mygrey} & Yes & Yes \\
\rowcolor{mygrey} \multirow{-2}{*}{Can all workers improve?} &  [\Cref{eg:MasterList-AllWorkersImprove}] & [\Cref{eg:MasterList-AllWorkersImprove}] \\
\multirow{2}{*}{Can some worker worsen?} & No  & No \\
& [\Cref{cor:no-worker-worse}]$^{\ddag}$ &  [\Cref{cor:no-worker-worse}]$^{\ddag}$ 
%
\end{tabular}
\vspace{0.05in}
\caption{The effect of one firm increasing its capacity  by $1$ on itself and the workers, under the worker-proposing (\WPDA{}) and firm-proposing (\FPDA{}) algorithms. Examples marked with $\S$ are from \cite{S97manipulation} and results marked with $\ddag$ are due to~\cite{GS85some} and~\cite{RS90two}.}
\vspace{-0.1in}
\label{tab:Summary}
\end{table}

The trends for capacity \emph{decrease} by a firm can be readily inferred from~\Cref{tab:Summary}. In particular, if increasing capacity can improve the firm's outcome, then going back from the new to the old instance implies that decreasing capacity makes it worse off.

One might intuitively expect that a firm should improve upon increasing its capacity, as it can now get matched with a strict superset of workers. Similarly, it is natural to think that an increase in a firm's capacity can also make some workers better off because an extra seat at a more preferable firm can allow some worker to switch to that firm, opening up the space for some other interested worker and so on, thus initiating a chain of improvements. \Cref{eg:MasterList-AllWorkersImprove} confirms this intuition on an instance where the workers' preferences are identical, also known as the \emph{master list} setting.

\begin{restatable}[All workers can improve]{example}{MasterList-AllWorkersImprove}
Consider an instance $\I$ with two firms $f_1,f_2$ and two workers $w_1,w_2$. The firm $f_1$ initially has zero capacity, while the firm $f_2$ has capacity $1$ (i.e., $c_1 = 0$ and $c_2 = 1$). Both workers have the preference $f_1 \> f_2 \> \emptyset$, and both firms have the preference $w_1 \> w_2 \> \emptyset$.
\begin{align*}
    w_1,w_2 : f_1 \succ f_2 \succ \emptyset \hspace{1in} f_1, f_2 : w_1 \succ w_2 \succ \emptyset
\end{align*}
The unique stable matching for this instance is $\mu_1 = \{(w_1,f_2)\}$.

Now consider a new instance $\I'$ obtained by adding unit capacity to firm $f_1$ (i.e., $c'_1 = 1$). The instance $\I'$ has a unique stable matching $\mu_2 = \{(w_1,f_1),(w_2,f_2)\}$. Observe that the workers $w_1$ and $w_2$ and the firm $f_1$ that increased its capacity are better off under the new matching $\mu_2$. Furthermore, as there is only one stable matching, the said trend holds under both \FPDA{} and \WPDA{} algorithms. Also, note that the two sets of stable matchings are disjoint. Thus, no matching is simultaneously stable for both old and new instances.\qed
\label{eg:MasterList-AllWorkersImprove}
\end{restatable}

Somewhat surprisingly, it turns out that increasing capacity can also \emph{worsen} a firm. This observation follows from the construction of \citet{S97manipulation}, who showed that any stable matching algorithm is vulnerable to manipulation via underreporting of capacity by some firm. We recall \citeauthor{S97manipulation}'s construction in \Cref{example:LP-firm-worse} below.

Intuitively, when workers propose under the \WPDA{} algorithm, a firm can worsen upon capacity increase (equivalently, improve upon capacity decrease) because of the following reason: By having \emph{fewer} seats, and thus by being \emph{more} selective, the firm can initiate rejection chains which may prompt more preferable workers to propose to it. On the other hand, by adding an extra seat, a firm may be forced to accept a suboptimal set of workers. This reasoning underlies the manipulation in \Cref{example:LP-firm-worse}.

Similar reasoning works when the firms propose under the \FPDA{} algorithm: With extra seats available, a firm may be \emph{forced} to make additional proposals to less-preferred workers. This could lead to rejection chains, causing other firms to hire its more preferred workers. Again, this phenomenon is at play in \Cref{example:LP-firm-worse}.

\begin{restatable}[Increasing capacity can worsen a firm under lexicographic preferences;~\citealp{S97manipulation}]{example}{LP-firm-worse}
Consider an instance $\I$ with three workers $w_1,w_2,w_3$ and two firms $f_1,f_2$ with lexicographic preferences given by
\begin{align*}
    w_1 &: f_2 \succ f_1 \succ \emptyset & \hspace{0.6in} f_1 &: w_1 \succ w_2 \succ w_3 \succ \emptyset \\    
    w_2, w_3  &: f_1\succ f_2 \succ \emptyset &\hspace{0.6in} f_2 &: w_3 \succ w_2 \succ w_1 \succ \emptyset
\end{align*}
%
%
Due to the lexicographic assumption, the firms' preferences over sets of workers are given by:
\begin{align*}
    f_1 &: \{w_1,w_2,w_3\} \succ \{w_1, w_2\} \succ \{w_1,w_3\} \succ 
    \{w_1\} \succ \{w_2,w_3\} \succ \{w_2\} \succ \{w_3\} \succ \emptyset\\
    f_2 &: \{w_1,w_2,w_3\} \succ \{w_2, w_3\} \succ \{w_1,w_3\} \succ 
    \{w_3\} \succ \{w_1,w_2\} \succ \{w_2\} \succ \{w_1\} \succ \emptyset
\end{align*}

Initially, each firm has unit capacity, i.e., $c_1 = c_2 = 1$. In this case, there is a unique stable matching, namely 
\[\mu_1=\{(w_1,f_1),(w_3,f_2)\}.\]
Now consider a new instance $\I'$ derived from the instance $\I$ by increasing the capacity of firm $f_1$ by $1$ (i.e., $c'_1=2$ and $c'_2=1$). The stable matchings for the instance $\I'$ are
\begin{align*}
    \mu_2 &= \{(\{w_1,w_2\},f_1),(w_3,f_2)\} \text{ and }\\
    \mu_3 &= \{(\{w_2,w_3\},f_1),(w_1,f_2)\}.
\end{align*}
Here, the firm-optimal stable matching (FOSM) is $\mu_2$, and the worker-optimal stable matching (WOSM) is $\mu_3$.

Finally, consider another instance $\I''$ derived from $\I'$ by increasing the capacity of firm $f_2$ by $1$ (i.e., $c''_1=2$ and $c''_2 = 2$). The unique stable matching for the instance $\I''$ is $\mu_3$. 

As a result of being the unique stable matching, the matching $\mu_1$ is FOSM and WOSM for the instance $\I$, and the matching $\mu_3$ is FOSM and WOSM for the instance $\I''$. Observe that firm $f_1$ prefers $\mu_1$ over $\mu_3$. Thus, under the \WPDA{} algorithm, the transition from $\I$ to $\I'$ exemplifies that a firm (namely, $f_1$) can worsen upon increasing its capacity. Similarly, the firm $f_2$ prefers $\mu_2$ over $\mu_3$. Thus, under the \FPDA{} algorithm, the transition from $\I'$ to $\I''$ exemplifies that a firm (namely, $f_2$) can worsen upon increasing its capacity.\qed
\label{example:LP-firm-worse}
\end{restatable}

Note that \Cref{example:LP-firm-worse} crucially uses the lexicographic preference structure; indeed, firm $f_1$ prefers being matched with the solitary worker $\{w_1\}$ over being assigned the pair $\{w_2,w_3\}$. One might ask whether the implication of \Cref{example:LP-firm-worse} holds in the absence of the lexicographic assumption. \Cref{prop:nofirmworse}, due to \citet{KU06games}, 
shows that under \emph{strongly monotone} preferences and the \WPDA{} algorithm, a firm \emph{cannot} worsen upon capacity increase.

\begin{restatable}[\citealp{KU06games}]{proposition}{nofirmworse}
Let $\mu$ and $\mu'$ denote the worker-optimal stable matching before and after a firm $f$ with strongly monotone preferences increases its capacity by $1$. Then, $\mu'(f) \succeq_{f} \mu(f)$.
\label{prop:nofirmworse}
\end{restatable}

The detailed proof of \Cref{prop:nofirmworse} is presented in \Cref{subsec:Capacity_Trends_Appendix}, but the main idea is as follows: Under the \WPDA{} algorithm, it can be shown that if the number of workers matched with a firm $f$ does not change upon capacity increase, then the \emph{set} of workers matched with $f$ also remains the same. (Notably, this observation does not require the preferences to be strongly monotone.) It can also be shown that the number of workers matched with firm $f$ cannot decrease upon capacity increase. (Again, this observation does not require strong monotonicity.). Thus, for the firm's outcome to change, it must be matched with strictly more workers in the new matching. Strong monotonicity then implies that the firm must strictly prefer the new outcome.

In contrast to \WPDA{}, a firm \emph{can} worsen upon capacity increase under the \FPDA{} algorithm even under strongly monotone preferences~(\Cref{example:increase-firm-worse}).
\begin{restatable}[Increasing capacity can worsen a firm under strongly monotone preferences;~\citealp{S97manipulation}]{example}{SM-firm-worse}
Consider the following instance, with two workers $w_1,w_2$ and two firms $f_1,f_2$ with strongly monotone preferences:
\begin{align*}
    w_1 &: f_2 \succ f_1 \succ \emptyset & \hspace{1cm} f_1 &:  \{w_1, w_2\} \succ \{w_1\} \succ \{w_2\} \succ \emptyset\\
    w_2 &: f_1\succ f_2 \succ \emptyset & \hspace{1cm} f_2 &: \{w_1,w_2\} \succ \{w_2\} \succ \{w_1\} \succ \emptyset
\end{align*}
Initially, each firm has unit capacity, i.e., $c_1 = c_2 = 1$. In this case, the firm-optimal stable matching is 
\[\mu_1=\{(w_1,f_1),(w_2,f_2)\}.\]
Upon increasing the capacity of firm $f_2$ to $c_2=2$ while keeping $c_1=1$, the firm-optimal stable matching of the new instance is
\begin{align*}
    \mu_2 &= \{(w_1,f_2),(w_2,f_1)\},
\end{align*}
which is worse for firm $f_2$ compared to the old matching $\mu_1$.\qed
\label{example:increase-firm-worse}
\end{restatable}

Finally, we note that under both \FPDA{} and \WPDA{} algorithms, no worker's outcome can worsen when a firm increases its capacity. The reason is that increasing the capacity of a firm corresponds to ``adding a man'' in the corresponding canonical one-to-one instance. Due to the increased ``competition'' among men, the outcomes of all women weakly improve (\Cref{prop:no-women-worse-wosm}). 

\begin{restatable}[\citealp{GS85some,RS90two}]{proposition}{no-women-worse-wosm}
Given any one-to-one instance $\I = \langle P, Q, \> \rangle$, let $\I' = \langle P \cup \{p\}, Q, \>' \rangle$ be another one-to-one instance derived from $\I$ by adding the man $p$ such that the new preferences~$\>'$ agree with the old preferences $\>$ on $P$ and $Q$. Let $\mu_P$ and $\mu_Q$ be the men-optimal and women-optimal stable matchings, respectively, for $\I$, and let $\mu'_P$ and $\mu'_Q$ denote the same for $\I'$. Then, for every woman $q \in Q$, we have $\mu'_P(q) \succeq'_{q} \mu_P(q)$ and $\mu'_Q(q) \succeq'_{q} \mu_Q(q)$. Furthermore, for every man $m \in P$, $\mu_P \succeq_m \mu'_P$ and $\mu_Q \succeq_m \mu'_Q$.
\label{prop:no-women-worse-wosm}
\end{restatable}

Using~\Cref{prop:no-women-worse-wosm} on the canonical one-to-one instance, we obtain that increasing a firm's capacity can never worsen the outcome of any worker under either worker-optimal or firm-optimal stable matching.

\begin{restatable}{corollary}{}
Let $\mu_W$ and $\mu_W'$ denote the worker-optimal stable matching before and after a firm $f$ increases its capacity by $1$, and let $\mu_F$ and $\mu_F'$ be the corresponding firm-optimal matchings. Then, for all workers $w \in W$, $\mu'_W(w) \succeq_{w} \mu_W(w)$ and $\mu'_F(w) \succeq_{w} \mu_F(w)$. For all firms $f' \neq f$, we have $\mu_W(f') \succeq_{f'} \mu'_W(f')$ and $\mu_F(f') \succeq_{f'} \mu'_F(f')$.
\label{cor:no-worker-worse}
\end{restatable}

Another implication of \Cref{prop:no-women-worse-wosm} is that adding a man makes every man in the set $P$ weakly worse off under both the new men-optimal and the new women-optimal matchings. 
Since men correspond to firms under the canonical transformation, one might wonder if this implication conflicts with our earlier observation from \Cref{eg:MasterList-AllWorkersImprove} that ``a firm can improve upon increasing its capacity''. We note that these observations are not in conflict since the preference of a firm depends on the \emph{set} of workers matched with \emph{all} of its copies. If the man corresponding to the newly added seat for a firm gets matched in the new instance while the men corresponding to the firm's existing seats retain their original partners, then, overall, the firm is strictly better off. This is the case in \Cref{eg:MasterList-AllWorkersImprove} for firm $f_1$. However, it is also possible for a firm to be worse off despite its newly added seat getting matched, as we saw in \Cref{example:LP-firm-worse}.

\section{Computational Results}\label{sec:Computational_Results}

In this section, we will examine the algorithmic aspects of capacity modification. We will approach this from the perspective of a central planner with the ability to adjust the capacities of firms to achieve a specific goal. This issue may arise, for example, when a large corporation is making adjustments within its various divisions or subsidiaries, or when a group of firms collaborates to make internal changes across the board in pursuit of a common objective.

We will focus on two natural and mutually incomparable objectives:

\begin{enumerate}
    \item \emph{Match a pair $(f^*,w^*)$}, where the goal is to determine if a fixed firm $f^*$ and a fixed worker $w^*$ can be matched under \emph{some} stable matching in the modified instance, and 
    \item \emph{stabilize a matching $\mu^*$}, where the goal is to check if a given matching $\mu^*$ or its subset can be realized as a stable outcome of the modified instance.
\end{enumerate}

These objectives have previously been studied in the one-to-one stable matching problem motivated by \emph{control} problems~\citep{BBH+21bribery,GJ23manipulation}.

We will assume that the central planner can modify the firms' capacities in one of the following two natural ways: (1) By \emph{adding} capacity, wherein the firms can receive some extra seats (the distribution can be unequal), and (2) by \emph{deleting} capacity, wherein some of the existing seats can be removed. Under both addition and deletion problems, we will assume that there is a \emph{global budget} $\ell \in \N$ that specifies the maximum number of seats that can be added (or removed) in aggregate across all firms.

The two objectives (match the pair and stabilize) and two actions (add and delete) give rise to four computational problems. One of these problems---adding capacity to match a pair---is formally defined below.

\medskip
\begin{center}
	{\small 
		\begin{tabularx}{1.0\columnwidth}{ll}
			\toprule
			\multicolumn{2}{c}{\textsc{Add Capacity To Match Pair}} \\
			\midrule
			\textbf{Given:}& \parbox[t]{0.75\columnwidth}{
			 An instance $\I = \langle F, W ,C, \> \rangle$, a worker-firm pair $(w^*,f^*)$, and a global budget $\ell \in \whole$.
				\vspace*{1mm}} \\%
			\textbf{Question:}& \parbox[t]{0.75\columnwidth}{
				Does there exist a capacity vector $\overline{C} \in (\whole)^n$ such that $\overline{C} \geq C$, $|\overline{C} - C|_{1} \leq \ell$, and $f^*$ and $w^*$ are matched in some stable matching of the instance $\I' = \langle F,W,\overline{C},\>\rangle$?} \\ 
			\bottomrule
		\end{tabularx}
	}
\end{center}
\medskip

Observe that the remaining three problems can be defined analogously: given a fixed budget of seats that may be added or removed, is it possible to match a given pair or stabilize a given matching. For example, let us formally describe the problem involving capacity deletion with the objective of stabilizing a given matching.

\problembox{\UnrestrictedDeleteCapacityStabilizeMatching{}}{
An instance $\I = \langle F,W, C, \> \rangle$, a global budget $\ell \in \whole$, and a matching $\mu^*$ with $|\mu^*(f)| \leq c_f$. 
}{
Does there exist a capacity vector $\overline{C} \in (\whole)^n$ such that $\overline{C} \leq C, |\overline{C} - C|_1 \leq \ell$ and there is a matching $\mu$ stable for $\I' = \langle F, W, \overline{C}, \> \rangle$ which is a subset of $\mu^*$, i.e., $\mu(f) \subseteq \mu^*(f)$ for each $f$ and $\mu(w) = \mu^*(w)$ or $\mu(w) = \emptyset$ for each $w$ ?
}

The problems described above allow the global budget $\ell$ to be spent as needed without restricting the minimum or maximum capacity. A natural generalization is to consider \emph{individual} budgets for the firms. For example, in the add capacity problem, in addition to the global budget $\ell$, we can also have an individual budget $\ell_f$ for each firm $f$ specifying the maximum number of additional seats that can be given to firm $f$. We call this generalization the \emph{budgeted} version and use the term \emph{unbudgeted} to refer to the problem with only global---but not individual---budget. Formally, the budgeted version of \textsc{Add Capacity to Match Pair} problem is defined as follows:

\medskip
\begin{center}
	{\small 
		\begin{tabularx}{1.0\columnwidth}{ll}
			\toprule
			\multicolumn{2}{c}{\textsc{Budgeted Add Capacity To Match Pair}} \\
			\midrule
			\textbf{Given:}& \parbox[t]{0.75\columnwidth}{
			 An instance $\I = \langle F, W ,C, \> \rangle$, a worker-firm pair $(w^*,f^*)$, a global budget $\ell \in \whole$, and an individual budget $\ell_f \in \whole$ for each firm $f$.
				\vspace*{1mm}} \\%
			\textbf{Question:}& \parbox[t]{0.75\columnwidth}{
				Does there exist a capacity vector $\overline{C} \in (\whole)^n$ such that $\overline{C} \geq C$, $|\overline{C} - C|_{1} \leq \ell$, $|\overline{c}_f - c_f| \leq \ell_f$ for each firm $f$, and $f^*$ and $w^*$ are matched in some stable matching of the instance $\I' = \langle F,W,\overline{C},\>\rangle$?} \\ 
			\bottomrule
		\end{tabularx}
	}
\end{center}
\medskip

The presence or absence of individual budgets on the previously defined objectives results in eight computational problems overall. \Cref{tab:manipulation} summarizes our results on the computational complexity of these problems.

\begin{table*}[t]
    \centering
    \begin{tabular}{@{\extracolsep{15pt}}c c c c c@{}}\\
    %
    %
    {} & \multicolumn{2}{c}{\textbf{Match the pair $(f^*, w^*)$}} & \multicolumn{2}{c}{\textbf{Stabilize the matching $\mu^*$}} \\
    \cmidrule{2-3}
    \cmidrule{4-5}
    {} & Add Capacity & Delete Capacity & Add Capacity & Delete Capacity\\
    %
    \cmidrule{2-5}
    %
    %
    \multirow{2}{*}{\textbf{Unbudgeted}} & \Polytime{} & \Polytime{} & \Polytime{} & \Polytime{}\\
    & [\Cref{theorem:Add-Capacity-Match-Pair-Unbudgeted}] & [\Cref{theorem:Delete-Capacity-Pair-Polynomial}] & [\Cref{theorem:UnrestrictedAddCapacityStabilize-Polynomial}] & [\Cref{theorem:Unrestricted-Delete-Capacity-Matching-Polynomial}]\\    
    \arrayrulecolor{white} 
    \cmidrule{2-5}
    \arrayrulecolor{black}
    %
    %
    \multirow{2}{*}{\textbf{Budgeted}}  & \NPH{} 
    & \NPH{}
    & \Polytime{} & \Polytime{} \\
    & [\Cref{theorem:Add-Capacity-to-Match-Pair-W1-Hard}] & [\Cref{theorem:Delete-Capacity-W1-Hard}] & [\Cref{theorem:RestrictedAddCapacityStabilizeMatching-Polynomial}] & [\Cref{theorem:Delete-Capacity-Matching-Polynomial}]
    \vspace{0.05in}
    \end{tabular}
    \caption{Summary of our computational results for adding and deleting capacity under two problems: matching a worker-firm pair (columns 2 and 3) and stabilizing a given matching (columns 4 and 5). The top row contains the results for the \emph{unbudgeted} problem (when only the aggregate change in firms' capacities is constrained), while the bottom row corresponds to the \emph{budgeted} problem (with additional constraints on individual firms). }
    \vspace{-0.05in}
    \label{tab:manipulation}
\end{table*}

A special case of the budgeted/unbudgeted problems is when the global budget is zero, i.e., $\ell = 0$. In this case, the capacities of the firms cannot be changed, and the goal is simply to check whether a worker-firm pair $(w^*,f^*)$ is matched in some stable matching for the original instance $\I$, or whether a given matching $\mu^*$ is stable for $\I$. The latter problem is straightforward. To solve the former problem, it is helpful to consider the canonical one-to-one instance of the given instance $\I$. For the one-to-one stable matching problem, a polynomial-time algorithm is known for listing all man-woman pairs that are matched in one or more stable matchings~\citep{G87three}. Using the bijection between the stable matchings of the two instances~(\Cref{prop:Canonical}), we obtain an algorithm to check if the worker $w^*$ is matched with any copy of firm $f^*$ in any stable matching.

Thus, the zero budget case can be efficiently solved for all the problems mentioned above. In the remainder of the section, we will consider the case of strictly positive global budgets.

We will now discuss the unbudgeted and budgeted versions of the \textsc{Add Capacity to Match Pair} and the \textsc{Delete Capacity to Stabilize Matching} problems. The detailed presentation of other problems listed in \Cref{tab:manipulation} can be found in \Cref{sec:Appendix_Computational_Results}.

\subsection*{Adding Capacity to Match A Pair: Unbudgeted}

Let us start with the problem of adding capacity to match a worker-firm pair $(w^*,f^*)$ in the unbudgeted setting, i.e., with global but without individual budgets. We find that this problem can be solved in polynomial time via \Cref{alg:globaladdcapacity}.  Although some of the problems in this section can be solved by a canonical reduction to a one-to-one instance, we cannot do so for \UnrestrictedAddCapacityMatchPair{}. 

The problem of adding at most $\ell$ men from a set of ``addable" men is shown to be \NPH{} by \citet{BBH+21bribery}. The reason this intractability does not extend to \UnrestrictedAddCapacityMatchPair{} is that under the unbudgeted version, we are allowed to add $\ell$ copies of a firm, which is not necessary under the version of the problem studied by \citet{BBH+21bribery}. In fact, our algorithm works by adding all the needed additional capacity to the same agent whenever possible. 

\paragraph{Algorithm overview.} To determine if the given worker-firm pair $(w^*,f^*)$ can be matched in some stable matching in the given instance $\I$ by increasing the capacity of the firms, our algorithm (see \Cref{alg:globaladdcapacity}) considers a modified instance $\I'$. In this new instance, $w^*$ and $f^*$ are already matched, and we check if it is possible to construct a stable matching for the remaining workers and firms while satisfying additional conditions.

\begin{algorithm}[ht]
\DontPrintSemicolon
  \KwIn{An instance $\I = \langle F, W, C, \> \rangle$, budget $\ell$, firm-worker pair $(f^*,w^*)$}
  \KwOut{A matching $\mu^*$ and a capacity vector $C^*$}
  Initialize the set of distracting workers $DW \gets \{ w \in W \, : \, w\succ_{f^*} w^*\}$ \;
  Initialize set of distracting firms $DF \gets \{ f \in F \, : \, f\succ_{w^*} f^*\}$ \; 
  Create $\I' = \langle F', W', C', \>' \rangle$ from $\I$ as follows:\\
  ~~ $F'\gets F$ and $W'\gets W\setminus \{w^*\}$\;
  ~~ $c'_f\gets c_f$ for all $f\neq f^*$ and $c'_{f^*}\gets c_{f^*}-1$\;
  ~~ Preferences $\>'\gets \>$ (after removing $w^*$)\label{line:Truncated-Prefs-1}\;
  ~~ For all $w\in DW$, delete acceptability of all firms   $f$ s.t. $f^*\>_w f$ under $\>_w'$\label{line:Truncated-Prefs-2}\;
  ~~ For all $f\in DF$, delete acceptability of all workers $w$ s.t. $w^*\>_f w$ under $\>_f'$\label{line:Truncated-Prefs-3}\;

  Choose the worker-optimal stable matching $\mu'$ for $\I'$\;
  
  Let $UDW\gets \{w\in DW \, : \, w$ is unmatched under   $ \mu'\}$\;
  Let $UDF\gets \{f\in DF \, : \, f$ is unsaturated under $ \mu'\}$\;
  
  \eIf{$UDF\neq \emptyset$ \textbf{OR} $|UDW|>\ell$\label{line:Unsaturated_Firm_Worker_Bound}}{
        $\mu^* \gets \emptyset$ and $C^*\gets C$\;
  }
  {
        $\mu^* \gets \mu' \cup (\{w^*\},f^*) \cup\{(w,f^*)|w\in UDW\} $\; 
        For all $f\neq f^*$, $c^*_f\gets c_f$\;
        $c^*_{f^*}\gets c_{f^*}+|UDW|$\;
  }
  \textbf{Return} $\mu^*$ and $C^*$\;
   \caption{Add Capacity to Match Pair}
   \label{alg:globaladdcapacity}
\end{algorithm}

More concretely, the algorithm considers the set of firms $DF$ (short for ``distracting firms'') that the worker $w^*$ prefers more than the firm $f^*$, and the set of workers $DW$ (short for ``distracting workers'') that the firm $f^*$ prefers more than $w^*$. Note that once the worker $w^*$ is matched with the firm $f^*$, the firms in $DF$ are the only ones that could potentially form a blocking pair with $w^*$. 
Similarly, the workers in $DW$ are the only ones that can block with $f^*$ due to the forced assignment of $w^*$ to $f^*$. For $f^*$ and $w^*$ to be stably matched, all distracting agents (the members of $DF$ and $DW$) must weakly prefer their matches to $f^*$ and $w^*$.

The algorithm creates the modified instance $\I'$ by truncating the preference lists of the firms in $DF$ (respectively, the workers in $DW$) by having them declare all workers ranked below $w^*$ (respectively, all firms ranked below $f^*$) as unacceptable. The truncation step is motivated by the following observation: In the original instance $\I$, there is a stable matching that matches $(w^*,f^*)$ after adding capacities to the firms if and only if there exists a stable matching in the truncated instance $\I'$ such that, after the added capacities, all firms in the set $DF$ are \emph{saturated} (i.e., matched with workers they prefer more than $w^*$) and all workers in the set $DW$ are \emph{matched} (i.e., matched either with $f^*$ or with firms they prefer more than $f^*$).

The key observation in our proof is that the desired matching exists in the truncated instance $\I'$ after adding capacities to the firms if and only if there exists a stable matching in the instance $\I'$ when the \emph{entire} capacity budget is given to the firm $f^*$. This observation readily gives a polynomial-time algorithm. 

We will start by proving a lemma that establishes a necessary condition for a solution to exist. Recall the sets $DF$ and $UDF$ and the matching $\mu'$ defined in \Cref{alg:globaladdcapacity}. We show that whenever a solution exists, all firms in $DF$ are matched under $\mu'$.

\begin{lemma}
Suppose the given instance of \UnrestrictedAddCapacityMatchPair{} admits a solution. Then, all distracting firms are saturated in the matching $\mu'$ found in \Cref{alg:globaladdcapacity}, i.e., $UDF = \emptyset$.
\label{lem:No_Unsaturated_Firm}
\end{lemma}
\begin{proof}
We are given that there exists a way of adding at most $\ell$ seats in the instance $\I$ such that $w^*$ and $f^*$ get matched under some stable matching, say $\sigma$, of the new instance. 
Let $\Delta C \in (\mathbb{N} \cup \{0\})^n$ denote the vector of the associated increase in the firms' capacities, and let $\I^+ = \langle F, W, C + \Delta C, \> \rangle$ denote the new instance. Thus, $\sigma$ is a stable matching for $\I^+$. Since $w^*$ matches with $f^*$ in $\sigma$, any distracting firm must be saturated in $\sigma$ (otherwise, it would create a blocking pair with $w^*$).

Consider the instance $\I'$ constructed in \Cref{alg:globaladdcapacity}, and let the instance $\I'' = \langle F', W', C' + \Delta C, \>' \rangle$ be obtained from $\I'$ by increasing the capacities in $\I'$ by $\Delta C$. Note that the instance $\I''$ can also be obtained from the instance $\I^+$ (under which $\sigma$ is stable) by the same way $\I'$ is obtained from $\I$. That is, given $\I^+$ if we freeze the assignment of $w^*$ to $f^*$, reduce the capacity of $f^*$ by $1$, and truncate the preference lists of the distracting workers and firm, we obtain $\I''$. 

We have observed above that the matching $\sigma$ must be stable for $\I^+$. It follows that the matching $\sigma \setminus \{w^*,f^*\}$ must be stable with respect to $\I''$. Furthermore, as observed above, any distracting firm must be saturated in $\sigma$ in the instance $\I^+$. Therefore, any distracting firm 
must be saturated under the matching $\sigma \setminus \{w^*,f^*\}$ in $\I''$. We shall use this to show that under $\I'$ no stable matching can have unsaturated distracting firms.

We will prove the lemma by contradiction. Suppose there exists a firm $f \in UDF$. Note that $f \neq f^*$ by definition of distracting firms. We will show that firm $f$ is unsaturated under every stable matching in the instance $\I''$, which would give the desired contradiction.

To prove the above claim, observe that \Cref{cor:no-worker-worse} implies that adding a seat to a firm makes all other firms weakly worse off. In fact, a stronger form of \Cref{cor:no-worker-worse} follows from \Cref{prop:no-women-worse-wosm} under the canonical transformation where men correspond to firms: When a seat is added to a firm, then \emph{every seat} of every other firm gets a weakly worse match. Thus, adding seats to firms other than $f$ does not increase the number of workers matched with $f$. 

Furthermore, since $f$ is unsaturated, the number of proposals it receives under the worker-proposing algorithm equals the number of workers matched with it. Increasing the capacity of $f$ does not affect the set of proposals received by $f$; therefore, the set of workers matched with it remains unchanged. By the rural hospitals theorem (\Cref{prop:RuralHospitals}), firm $f$ is matched with the same set of workers under all stable matchings. Thus, firm $f$ is unsaturated in all stable matchings in the instance $\I''$. This contradicts with the fact that $f$ is saturated under $\sigma\setminus \{f^*,w^*\}$ under $\I''$, as desired.    
\end{proof}

We now utilize this result to prove the following theorem.

\begin{restatable}{theorem}{AddCapacityMatchPairUnbudgeted}
\UnrestrictedAddCapacityMatchPair{} can be solved in polynomial time.
\label{theorem:Add-Capacity-Match-Pair-Unbudgeted}
\end{restatable}

\begin{proof}
We will show that \Cref{alg:globaladdcapacity} solves \UnrestrictedAddCapacityMatchPair{} in polynomial time.

To prove correctness, we will show that the following conditions are necessary and sufficient for the existence of a solution to \UnrestrictedAddCapacityMatchPair{}:
\begin{itemize}
    \item There is no unsaturated distracting firm in the matching $\mu'$ (i.e., $UDF = \emptyset$), and
    \item there are at most $\ell$ unmatched distracting workers in $\mu'$  (i.e., $|UDW| \leq \ell$).
\end{itemize}
The correctness of the algorithm will follow as these checks are performed in Line~\ref{line:Unsaturated_Firm_Worker_Bound}.

We know from \Cref{lem:No_Unsaturated_Firm} that the $UDF = \emptyset$ condition is necessary. To see why $|UDW| \leq \ell$ is necessary, assume, on the contrary, that there are more than $\ell$ unmatched distracting workers in $\mu'$. Adding one seat to any firm can decrease the number of unmatched distracting workers by at most one. Therefore, if the number of workers in $UDW$ is greater than $\ell$, the given instance of \UnrestrictedAddCapacityMatchPair{} has no solution.

We will now argue that these conditions are also sufficient. Suppose $UDF=\emptyset$ and $|UDW|\leq \ell$. Then, all unmatched distracting workers can be assigned to $f^*$ by increasing its capacity by $|UDW|$. The resulting matching $\mu^*$ is stable with respect to the original preferences $\>$ under the modified capacities. Consequently, the correctness of the algorithm follows.

Finally, the running time guarantee follows by noting that the algorithm runs in time polynomial in the number of firms and workers.
\end{proof}

\subsection*{Adding Capacity to Match A Pair: Budgeted}

Next, we will consider a more general problem where, in addition to the global budget of $\ell$ seats, we are also given an individual budget $\ell_f$ for each firm $f$ specifying the maximum number of seats that can be added to the firm $f$. The goal, as before, is to determine if, after adding capacities as per the given budgets, it is possible to match the pair $(f^*,w^*)$ under some stable matching.

Note that our algorithm for the unbudgeted problem assigns the entire additional capacity to the firm $f^*$, which may no longer be feasible in the budgeted problem. It turns out that unless $P = NP$, there cannot be a polynomial-time algorithm for this problem.

\begin{restatable}{theorem}{AddCapacitytoMatchPairNPHard}
\RestrictedAddCapacityMatchPair{} is \NPH{}.
\label{theorem:Add-Capacity-to-Match-Pair-W1-Hard}
\end{restatable}

To prove \Cref{theorem:Add-Capacity-to-Match-Pair-W1-Hard}, we leverage a result of \citet{BBH+21bribery} on control problems in the one-to-one stable matching problem~(which, as per our convention, involves a matching between men and women). Specifically, \citet{BBH+21bribery} study the problem of adding a set of at most $\ell$ agents (men or women) such that a fixed man-woman pair is matched under some stable matching in the resulting instance.

Interestingly, the reduction of \citet{BBH+21bribery} holds even when only men (but not women) are added and the set of addable men is a fixed subset of the original men. Due to this additional feature, we redefine the problem of \citet{BBH+21bribery} and call it \ConstructiveExistsAddMen{}.

\smallskip
\begin{center}
	{\small 
		\begin{tabularx}{1.0\columnwidth}{ll}
			\toprule
			\multicolumn{2}{c}{\ConstructiveExistsAddMen{}} \\
			\midrule
			\textbf{Given:}& \parbox[t]{0.75\columnwidth}{
			An instance $\I = \langle P_{orig},Q,\> \rangle$, a set of addable men $P_{add} \subseteq P_{orig}$, 
            a man-woman pair $(p^{*}, q^{*})$ from the original set of agents, and a budget $\ell \in \whole$.
				\vspace*{1mm}} \\%
			\textbf{Question:}& \parbox[t]{0.75\columnwidth}{
				Does there exist a set $\overline{P} \subseteq P_{add}$ such that $|\overline{P}| \leq \ell$ and $(p^*, q^*)$ is part of at least one stable matching in $\langle P_{orig} \cup \overline{P}, Q, \> \rangle$?} \\ 
			\bottomrule
		\end{tabularx}
	}
\end{center}
\smallskip

\begin{restatable}[\citealp{BBH+21bribery}]{proposition}{Constructive-Exists-Add-Men-W1-Hard}
\ConstructiveExistsAddMen{} is \NPH{}.
\label{prop:Constructive-Exists-Add-Men-W1-Hard}
\end{restatable}

We now use \Cref{prop:Constructive-Exists-Add-Men-W1-Hard} to show intractability for \RestrictedAddCapacityMatchPair{} using the following straightforward construction: For each man in the set $P_{orig}$, we create a firm with capacity $1$ and individual budget $\ell_f=0$, while for each man in the addable set $P_{add}$, we create a firm with capacity $0$ and individual budget $\ell_f=1$. Adding a seat to an individual firm corresponds to adding the associated man. The equivalence between the solutions of the two problems now follows. The fact that only men are the addable agents ($P_{add}$) in \ConstructiveExistsAddMen{} allows us to give a reduction from \ConstructiveExistsAddMen{} to \RestrictedAddCapacityMatchPair{}.

\AddCapacitytoMatchPairNPHard*

\begin{proof} 
    Suppose the input to \ConstructiveExistsAddMen{} consists of the original instance $\I = \langle P_{orig}, Q, \> \rangle$, an addable set $P_{add}$ of men, a man-woman pair $(p^*, q^*)$, and a budget $\ell$. We construct an instance 
    of \RestrictedAddCapacityMatchPair{} as follows: Create a firm $f_{p}$ with capacity 1 and individual budget 0 for each man $p$ in $P_{orig}$, a firm $f_p$ with capacity 0 and individual budget 1 for each man $p$ in $P_{add}$, and a worker $w_q$ for each woman $q$ in $Q$. We keep the same preferences over the corresponding agents and the same global budget, and set $f^* = f_{p^*}$ and $w^* = w_{q^*}$.

    Note that adding a person in the original instance is equivalent to increasing the capacity of the corresponding firm and utilizing a unit amount of its individual budget in the reduced instance. Since \ConstructiveExistsAddMen{} is \NPH{} by \Cref{prop:Constructive-Exists-Add-Men-W1-Hard}, we conclude that \RestrictedAddCapacityMatchPair{} is also \NPH{}.
\end{proof}


\subsection*{Deleting Capacity To Stabilize A Matching: Unbudgeted and Budgeted}

We will now consider the problem of {\em deleting} the capacities of firms to ensure that a subset of a given matching becomes stable. We first assume the existence of firm-specific budgets. This problem is formally defined as follows:  

\problembox{\RestrictedDeleteCapacityStabilizeMatching{}}{
An instance $\I = \langle F, W, C, \> \rangle$, a global budget $\ell \in \whole$, individual budget $\ell_f \in \whole$ for each firm $f \in F$, and a matching $\mu^*$ with $|\mu^*(f)| \leq c_f$. 
}{
Does there exist a capacity vector $\overline{C} \in (\whole)^n$ such that $\overline{C} \leq C, |\overline{C} - C|_1 \leq \ell$ and $|\overline{c}_f - c_f| \leq \ell_f$ and there is a matching $\mu$ stable for $\I' = \langle F, W, \overline{C}, \> \rangle$ which is a subset of $\mu^*$, i.e., $\mu(f) \subseteq \mu^*(f)$ for each $f$ and $\mu(w) = \mu^*(w)$ or $\mu(w) = \emptyset$ for each $w$ ?
}

We shall prove that this problem can be solved in polynomial time by applying the canonical reduction and providing an algorithm for its one-to-one equivalent.

\begin{restatable}{theorem}{DelCapStabMatch}
\RestrictedDeleteCapacityStabilizeMatching{} can be solved in polynomial time.
\label{theorem:Delete-Capacity-Matching-Polynomial}
\end{restatable}

Given this result, it is straightforward to see that the unbudgeted setting is solvable in polynomial time as well. We can create a budgeted instance from the unbudgeted one by simply setting the firm-specific budgets equal to the global budget. 

\begin{restatable}{theorem}{Unrestricted-Delete-Capacity-Matching-Polynomial}
\UnrestrictedDeleteCapacityStabilizeMatching{} can be solved in polynomial time.
\label{theorem:Unrestricted-Delete-Capacity-Matching-Polynomial}
\end{restatable}

We shall now turn our attention to the budgeted case. Interestingly, the one-to-one equivalent of this (budgeted) problem \ExactExistsDeleteMen{} can be solved in polynomial time. In contrast, the problem studied by \citet{BBH+21bribery}, where the set of agents that can be deleted includes either men or women, is \NPH{}. 

We will now show that the problem of \RestrictedDeleteCapacityStabilizeMatching{} can be solved in polynomial time. 
We first formally define \ExactExistsDeleteMen{}. 
In this problem, we wish to stabilize a given matching by deleting men. Here, the specific budgets come into play with all the men being partition into mutually disjoint groups, and each group has a budget on how many men from this group can be removed. 

\problembox{\ExactExistsDeleteMen{}}{A one-to-one instance $\I = \langle P , Q, \>  \rangle$ with a partition $(P_1,\cdots,P_k)$ on the set of men $P$, a global budget $\ell$, group-specific budgets $\ell_j$ for all men corresponding in the group $P_j$, and a matching $\mu^{*}$ defined on $\I$.}
{Does there exist a set of men $\overline{P} \subseteq P$ such that  $|\overline{P}| \leq \ell$, $\overline{P}$ has at most $\ell_j$ men from each $P_j$ and the matching $\mu^{*}$ when restricted to $P \cup Q \setminus \overline{P}$ is stable. More formally, does there exist a stable matching $\mu$ in $\I \setminus \overline{P}$ such that for each $p \in P \setminus \overline{P}$, $\mu(p) = \mu^*(p)$ and for each woman $q \in Q$, $\mu(q) = \mu^{*}(q)$ if $ \mu^{*}(q) \in P \setminus \overline{P}$ otherwise $\mu(q) = \emptyset$?}

We now show in \Cref{EEDM}  that \ExactExistsDeleteMen{} can be solved in polynomial time \footnote{Recall that this is different from the \ExactExistsDelete{} problem studied by \cite{BBH+21bribery} which is computationally intractable.}. \Cref{EEDM} proceeds by repeatedly deleting all the men that form a blocking pair till $\mu^{*}$ becomes stable. The idea is similar to the exponential-time algorithm for \ExactExistsDelete{} given in \cite[Algorithm 2]{BBH+21bribery}. \ExactExistsDelete{} considers deleting either men or women, but in our case we only wish to delete men, and additionally have a budgets on how many men can be deleted from a given group of men. This change makes the time complexity polynomial, as we will later prove.

\begin{algorithm}[t]
    \DontPrintSemicolon
  \KwIn{Instance $ \I = \langle P, Q, \> , (P_1,\cdots,P_k), \ell,$ $ \{\ell_j : j\in [k]\}, \mu^{*} \rangle$ }
  \KwOut{A set $\overline{P}$ of men and a matching $\mu \subseteq \mu^*$.}
   Initialize $\overline{P} \gets \emptyset$. \;
   Initialize $B \gets $ the set of men involved in a blocking pair of $\mu^*$ restricted to $\I$. \;

   \While{$B$ is not empty}{

    \ForAll{$p \in B$}
    {
        Let $j$ be s.t. $p\in P_j$
        Update $\ell \gets \ell -1$ \;
        Update $\ell_j \gets \ell_j -1$ \;
        Update $\overline{P} \gets \overline{P} \cup \{ p\}$
    }
    Update $B \gets$ the set of men involved in a blocking pair of $\mu^*$ restricted to $\I \setminus \overline{P}$. \;
   
   }\label{step:checkfeasible}
   \If{$\ell < 0$ OR $\ell_f < 0$ for some $f$ in $F$}
   {
        Set $\overline{P}\gets \emptyset$ and $\mu\gets\emptyset$ \Comment*[r]{Cannot stabilize $\mu^*$ within the budget.}
   }
    \textbf{Return} $\overline{P}$ and $\mu$\;
  
\caption{\ExactExistsDeleteMen{} algorithm}

\label{EEDM}
\end{algorithm}

To prove the correctness of the algorithm, we need the following lemma:

\begin{lemma}
Let $\I = \langle P, Q, \> \rangle$ be a stable matching instance, $P_{A}\subseteq P$ be a subset of men and $\mu^{*}$ be a matching in $\I$. Suppose $(p,q)$ is a blocking pair for $\mu^{*}$ restricted to $\I \setminus P_A$. Then, for any $P_B \subseteq P\setminus\{p\}$, $(p,q)$ is still a blocking pair for $\mu^{*}$ restricted to $\I \setminus (P_A \cup P_B)$.
\label{lemma:Blocking-Lemma}
\end{lemma}

\begin{proof} (of \Cref{lemma:Blocking-Lemma})
    Let $\mu$ be $\mu^{*}$ restricted to $\I \setminus P_A$. Since, $(p,q)$ is a blocking pair for $\mu$, we get that $q \>_{p} \mu(p)$ and $p \>_{q} \mu(q)$.

    Let $\mu'$ be $\mu^{*}$ restricted to $\I \setminus (P_A \cup P_B)$. Note that $P_A \cup P_B$ does not contain $p$, which means that $p$ and $q$ are still in the instance. Clearly $\mu' \subseteq \mu \subseteq \mu^{*}$. Therefore, $\mu'(a) = \mu(a)$ or $\mu'(a) = \emptyset$ for $a \in \{p ,q \}$. In both cases we get that $\mu(p) \succeq_p \mu'(q)$ and $\mu(q) \succeq_q \mu'(q)$. Hence, we get that $q \>_{p} \mu'(p)$ and $p \>_{q} \mu'(p)$. Consequently $(p,q)$ is a blocking pair in $\mu'$ and we are done.
\end{proof}

As a result, in order to stabilize a given matching $\mu$, we need to delete all men who form blocking pairs. We use this idea in our algorithm for \ExactExistsDeleteMen{}.

\begin{lemma}
\ExactExistsDeleteMen{} can be solved in polynomial time.
\label{lemma:Exact-Exists-Delete-Men-Polynomial}
\end{lemma}

The case where the algorithm returns the set of men $\overline{P}$ is fairly straightforward since the loop terminates when no blocking pair is remaining. As we explicitly check that the budgets are not violated, the set returned is guaranteed to be a valid solution. For the case with no solution, we show by induction that the elements that are contained in set $\overline{P}$ computed by our algorithm must also be present in any valid solution, and hence another set cannot satisfy the budget constraints.

\begin{proof} (of \Cref{lemma:Exact-Exists-Delete-Men-Polynomial})
We will now show that \Cref{EEDM} solves any given instance $ \I = \langle P, Q, \> , (P_1,\cdots, P_k), \ell, \{\ell_j : j\in [k]\}, \mu^{*} \rangle$ of \ExactExistsDeleteMen{} in polynomial time. First we show the correctness of the algorithm. Consider the set $\overline{P}$ computed by the algorithm.

\paragraph{Correctness.} We first show that whenever the algorithm terminates by returning the set $\overline{P}$ and the stable matching, the given instance is solved by $\overline{P}$. For every man we add to $\overline{P}$, we decrease the global and group-specific budget by 1. Thus, at Line \ref{step:checkfeasible} of \Cref{EEDM} we check that $\overline{P}$ does not exceed the global and individual budgets. Also, the while loo only breaks when $B$ is empty, hence, the matching $\mu^{*}$ restricted to $\I \setminus \overline{P}$ is stable and the given instance is solved by $\overline{P}$.

It remains to show that whenever the algorithm terminates by returning the empty matching, the given instance of \ExactExistsDeleteMen{} has no solution. The algorithm terminates with an empty matching when $\overline{P}$ exceeds the global budget or the group-specific budgets. For contradiction, assume that the given instance of \ExactExistsDeleteMen{} is solved by a subset of men $P^{\dagger} \subseteq P$ satisfying the global and group-specific budgets such that $\mu^{*}$ is stable restricted to $\I \setminus P^{\dagger}$. We will prove by induction that $\overline{P} \subseteq P^{\dagger}$, which will contradict the fact that $P^{\dagger}$ is within the budget. For the given instance, let there be $t$ iterations of the while loop. Let $\overline{P}_i$ be the set of men added to $\overline{P}$ during the $i$th iteration of the while loop. 

We claim that $\overline{P}_i \subseteq P^{\dagger}$ for all $i \in [t]$. The base case consists of showing $\overline{P}_1 \subseteq P^{\dagger}$. For the sake of contradiction, assume that there is a man $p$ such that $p\in \overline{P}_1$ but  $p \notin P^{\dagger}$. Let $(p,q)$ be the blocking pair formed by $p$ for $\mu^{*}$ restricted to $\I$. Using \Cref{lemma:Blocking-Lemma} on $\I$ with $P_A = \emptyset$ and $P_B = P^{\dagger}$, we conclude that $(p,q)$ is still a blocking pair for $\mu^{*}$ restricted to $\I \setminus P^{\dagger}$, which contradicts the assumption that $p\notin P^{\dagger}$. Hence, no such man can exist. Consequently, $\overline{P}_1 \subseteq P^{\dagger}$.

For the inductive step, assume that $\overline{P}_i \subseteq P^{\dagger}$ for all $i \in [k]$ for some $1 \leq k < t$. By definition, $\overline{P}_{k+1}$ is the set of men involved in a blocking pair of $\mu^{*}$ restricted to $\I \setminus (\overline{P}_1 \cup \dots \cup \overline{P}_k)$. Analogous to the base case, for contradiction, assume that there is a man $p$ s.t.  $p\in \overline{P}_{k+1}$ but  $p \notin P^{\dagger}$. Let $(p,q)$ be the blocking pair formed by $p$ for $\mu^{*}$ restricted to $\I \setminus (\overline{P}_1 \cup \dots \cup \overline{P}_k)$. 

Again, using \Cref{lemma:Blocking-Lemma} on $\I$ with $P_A = (\overline{P}_1 \cup \dots \cup \overline{P}_k)$ and $P_B = P^{\dagger}$, we obtain that $(p,q)$ is still a blocking pair for $\mu^{*}$ restricted to $\I \setminus (\overline{P}_1 \cup \dots \cup \overline{P}_k \cup P^{\dagger}) = \I \setminus P^{\dagger}$, which is a contradiction. Hence,  no such man can exist. Thus, $\overline{P}_{k+1} \subseteq P^{\dagger}$.

Hence, by induction, we prove that $\overline{P}_i \subseteq P^{\dagger}$ for all $i \in [t]$, which implies $\overline{P} \subseteq P^{\dagger}$. As $\overline{P}$ exceeds at least one of the global budget or the group-specific budgets, so should $P^{\dagger}$. This is a contradiction to the fact that $P^{\dagger}$ solves the given instance, and we conclude that the given instance of \ExactExistsDeleteMen{} has no solution.


\paragraph{Running Time.} 
Observe that $B \cap \overline{P}$ is always $\emptyset$. This is because $B$ contains agents from $P \cup Q \setminus \overline{P}$. Hence, in each iteration of the while loop we add at least one man to $\overline{P}$ and the while loop terminates in at most $|P|$ iterations. The inner for loop terminates in $|B| \leq |P|$ iterations and blocking pairs can be computed in linear time. Hence, the algorithm terminates in polynomial time. This proves that \ExactExistsDeleteMen{} can be solved in polynomial time.
\end{proof}

We can now show that \RestrictedDeleteCapacityStabilizeMatching{} can be solved in polynomial time by giving a reduction to \ExactExistsDeleteMen{}.

\DelCapStabMatch*

\begin{proof}

We shall establish this result by giving a polynomial time  reduction to  \ExactExistsDeleteMen{} which can be solved in poly time by \Cref{EEDM}.
Let the given instance of \RestrictedDeleteCapacityStabilizeMatching{} be $\langle \I =  F, W, C, \> , \ell, \{\ell_f, f \in F\}, \mu^* \rangle$. We define an instance  of \ExactExistsDeleteMen{} as follows: 

Let $\I'$ be the canonical one-to-one reduction of $\I$. We define $\mu'$ to be the one-to-one extension of $\mu^*$ to the one-to-one instance $\I'$. We set $k=n$, essentially making a group for each firm, where $P_f$ contains the set of men corresponding to $f$. Accordingly, we set the budget of the group $P_f$ as $\ell_{P_f}=\ell_{f}$.

We can see that if the instance constructed in the reduction $\I'$, can be solved by deleting a set of men $\overline{P}$, then we can also stabilize $\mu^*$ in $\I$ by deleting the capacities of the corresponding firms. 
Similarly, if we can solve the original instance by deleting some capacities of firms, we can solve the modified instance by deleting the corresponding men.  Hence, \RestrictedDeleteCapacityStabilizeMatching{} reduces to \ExactExistsDeleteMen{}. Since \ExactExistsDeleteMen{} can be solved in polynomial time by \Cref{lemma:Exact-Exists-Delete-Men-Polynomial}, we conclude that \RestrictedDeleteCapacityStabilizeMatching{} is solvable in polynomial time too.
\end{proof}


    

\section{Capacity Modification v/s Preference Manipulation}
\label{sec:Capacity_vs_Preference}

So far, we have discussed qualitative (\Cref{sec:Capacity_Modification_Trends}) and computational (\Cref{sec:Computational_Results}) aspects of capacity modification from the perspective of a \emph{central planner}. We will now adopt the perspective of a \emph{firm} and compare the different manipulation actions available to it. 

Specifically, we will consider \emph{preference manipulation} (abbreviated as \texttt{\textup{Pref}}), wherein a firm can misreport its preference list without changing its capacity, and compare it with the two capacity modification actions we have already seen, namely \texttt{\textup{Add}} and \texttt{\textup{Delete}} capacity, wherein the firm can increase or decrease its capacity without changing its preferences. These actions are formally defined below.

\begin{itemize}
    \item $\texttt{\textup{Pref}}$: Under this action, a firm can report any \emph{permutation} of its acceptable workers without changing its capacity.\footnote{Manipulation via permutation has been studied by several works in the stable matching literature~\citep{TST01gale,KM09successful,KM10cheating,GIM16total,VG17manipulating,SDT21coalitional,HUV21accomplice}.} That is, if a firm $f$'s true preference is $\succ_f$, then $\succ'_f$ is a valid preference manipulation if for any worker $w$, $w\succ_f \emptyset$ if and only if $w \succ'_f \emptyset$.
    \item \texttt{\textup{Add}}/\texttt{\textup{Delete}}: Under \texttt{\textup{Add}} (respectively, \texttt{\textup{Delete}}), the firm $f$ strictly increases (respectively, decreases) its capacity $c_f$ by an arbitrary amount without changing its preferences.
\end{itemize}

Our goal is to examine which mode of manipulation---\texttt{\textup{Pref}}, \texttt{\textup{Add}}, or \texttt{\textup{Delete}}---is always (or sometimes) more beneficial for the firm compared to the others under the \FPDA{} and \WPDA{} algorithms.\footnote{Throughout, we will use the word \emph{algorithm} instead of \emph{mechanism} because we only consider one strategic agent (all other agents are truthful) and do not consider game-theoretic equilibria.}

At first glance, each manipulation action may seem to offer a distinctive ability to the firm: \texttt{\textup{Add}} allows the firm to either tentatively accept more proposals (under \WPDA{}) or make more proposals (under \FPDA{}), thus facilitating larger-sized (and possibly more preferable) matches. \texttt{\textup{Delete}} can allow a firm to be more selective, which, as we have seen in \Cref{sec:Capacity_Modification_Trends}, can be advantageous in certain situations. Finally, \texttt{\textup{Pref}} can allow a firm to trigger specific rejection chains, resulting in a potentially better set of workers. Given the unique advantage of each manipulation action, a systematic comparison among them is well-motivated.

We compare the manipulation actions under the two versions of the deferred acceptance algorithm, \WPDA{} and \FPDA{}, and focus on a fixed firm $f$. An action $X$ is said to \emph{outperform} action $Y$, where $X,Y \in \{\texttt{\textup{Pref}},\texttt{\textup{Add}},\texttt{\textup{Delete}}\}$, if there exists an instance such that the outcome for firm $f$ when it performs $X$ is strictly more preferable to it than that under $Y$.

\paragraph{Introducing ``Peak''.} An important insight from our analysis is that the usefulness of a manipulation action depends on a threshold on the firm's capacity, which we call its \emph{peak}. With the preferences of all agents fixed along with the capacities of the other firms, the peak of firm $f$ is the size of the largest set of workers matched to $f$ under any stable matching when $f$ is free to choose its capacity~$c_f \in \mathbb{N}$. 

Formally, given an instance $\I=\langle F,W,C, \succ \rangle$, a firm $f\in F$ and any $b \in \mathbb{N}$, let $\I^b=\langle F,W,(C_{-f},b), \succ \rangle$ denote the instance derived from $\I$ where the capacity of firm $f$ is changed from $c_f$ to $b$ (and no other changes are made); here, $C_{-f}$ denotes the capacities of firms other than $f$. Recall that the set of stable matchings for an instance $\I$ is denoted by $\S_{\I}$. The \emph{peak} $p_f$ for firm $f$ is defined as the size of the largest set of workers $f$ is matched with under any stable matching in the instance $\I^b$ for an arbitrary choice of $b$, i.e., 
\[p_f(\I) \coloneqq \max_{b \in \mathbb{N}, \, \mu \in \S_{\I^b}} |\mu (f)|.\]

\subsection{Some Observations about Peak}

Although the peak is the maximum size of the \emph{matched set} of a firm at any capacity, it is not equal to the maximum number of \emph{proposals} the firm receives under the \WPDA{} algorithm for an arbitrarily chosen capacity, as the following example illustrates.

\begin{example}[Peak $\neq$ maximum number of proposals made to a firm]\label{ex:PeakvsProposals}
     Consider an instance $\I$ with three firms $f_1, f_2, f_3$ and four workers $w_1,w_2,w_3, w_4$. The firms have unit capacities (i.e., $c_1=c_2=c_3=1$) and have responsive preferences given by
    \begin{align*}
        w_1  &: f_1 \succ f_2 \succ f_3\succ  \emptyset & \hspace {0.15in} & f_1: w_1 \succ w_2 \succ w_3 \succ w_4 \\
        w_2 &: f_1 \succ f_2 \succ f_3 \succ \emptyset & \hspace {0.15in} & f_2: w_2 \succ w_3 \succ w_1 \succ w_4 \\
        w_3 &: f_2 \succ f_1 \succ f_3 \succ \emptyset & \hspace {0.15in} & f_3: w_3 \succ w_4 \succ w_2 \succ w_1 \\
        w_4 &: f_3 \succ f_1 \succ f_2 \succ \emptyset & 
    \end{align*}
    Under the \WPDA{} algorithm, firm $f_1$ receives proposals from all four workers. 
    
    Upon increasing its capacity to $c_1 = 2$, firm $f_1$ receives only two proposals, from $w_1$ and $w_2$, and is matched with $\{w_1,w_2\}$. It is easy to verify that this is the largest set of workers $f_1$ can be matched with. Thus, the peak of $f_1$ is $2$, which is strictly smaller than the number of proposals received by it under the \WPDA{} algorithm in the original instance.\qed
\end{example}

In fact, we can show that under the \WPDA{} algorithm, the set of proposals a firm $f$ receives at capacity $s$ is a subset of the set of proposals received by $f$ at capacity $r$ if $r\leq s$. We defer the proof of this to \Cref{app:capvspref}.

\begin{restatable}{proposition}{ProposalsSubset}\label{prop:ProposalsSubset}
Given instance $\I$, a firm $f$ and two nonnegative integers $r$ and $s$ such that $r\leq s$, if a worker $w$ proposes to $f$ under the \WPDA{} on $\I^s$, then it must propose to $f$ under the \WPDA{} on $\I^r$. 
\end{restatable}

\begin{corollary}\label{cor:proposalsnumber}
    Given $r\leq s$, the number of proposals a firm $f$ receives under \WPDA{} on $\I^s$ is less than or equal to the number of proposals $f$ receives under the \WPDA{} on $\I^r$. 
\end{corollary}

Note that when the capacity of a firm is greater than or equal to its peak (i.e., $c_f \geq p_f$), it does not reject any proposal during the \WPDA{} algorithm. This is because in order to reject a proposal for the first time, the firm must be saturated at that stage of the algorithm; thus, it must already hold proposals from at least $p_f$ workers. Then, by increasing its capacity, the firm will be able to hold strictly more than $p_f$ proposals in the new instance, implying that the firm will eventually be matched with strictly more than $p_f$ workers. This contradicts the fact that the peak is $p_f$.

\begin{restatable}{proposition}{PeakMaxProposalsWPDA}
When a firm's capacity is greater than or equal to its peak, it does not reject any proposals under the \WPDA{} algorithm.
\label{prop:Peak_Max_Proposals_WPDA}
\end{restatable}

When a firm's capacity is greater than its peak (i.e., $c_f > p_f$), 
the set of workers matched with it in any stable matching is the same as that in the worker-optimal stable matching in the at-peak instance. This is because the set of proposals received by a firm under the \WPDA{} algorithm in any above-peak instance is the same as that in the at-peak instance (since, by \Cref{prop:Peak_Max_Proposals_WPDA}, the firm does not reject any such proposals). Thus, the firm is matched with the same set of workers in the worker-optimal stable matching of any above-peak instance. Furthermore, the firm is unsaturated in any such matching. Thus, by the rural hospitals theorem~(\Cref{prop:RuralHospitals}), the firm is matched with the same set of workers in every stable matching.

\begin{restatable}{proposition}{CovergeBeyondPeak}
When a firm's capacity is greater than its peak, the set of workers matched with the firm in \emph{any} stable matching is the same as its worker-optimal match in the at-peak instance. 
\label{lem:convergebeyondpeak}
\end{restatable}

Combining \Cref{cor:proposalsnumber}, \Cref{prop:Peak_Max_Proposals_WPDA} and \Cref{lem:convergebeyondpeak}, we get the following theorem showing the significance of peak for the \WPDA{}.

\begin{theorem}
    Given a many-to-one matchings instance $\I$, a firm $f$ and value $b\in [m]$, let $\mu$ be a stable matching under $\I^b$. We have that under the \WPDA{} on $\I^b$:
    \begin{enumerate}
        \item if $b<p_f$, the number of proposal $f$ receives is larger than $|\mu(f)|$,
        \item if $b=p_f$, the number of proposals $f$ receives is $|\mu(f)|$ and 
        \item if $b>p_f$, the number of proposals $f$ receives is smaller than $|\mu(f)|$.
    \end{enumerate}
\end{theorem}

\begin{figure}[t]
\centering

\tikzset{every picture/.style={line width=1pt}}  

    \begin{subfigure}[b]{0.1\linewidth}
	\centering
	\begin{tikzpicture}
		\node[] (1) at (0,2) {{\WPDA{}:}};
            \node[] (1) at (0,0) { };
	\end{tikzpicture}
    \end{subfigure}
    \begin{subfigure}[b]{0.28\linewidth}
	\centering
	\begin{tikzpicture}
		\tikzset{mynode/.style = {shape=circle,draw,inner sep=0pt,minimum size=25pt}} 
		\node[mynode] (1) at (1.5,1.5) {\footnotesize{Pref}};
		\node[mynode] (2) at (2.6,0) {\footnotesize{Del}};
		\node[mynode] (3) at (0.4,0) {\footnotesize{Add}};
		    \draw[->,-latex] (1) to[out=-10,in=90] (2);
            \draw[->,-latex] (2) -- (1);
            \draw[->,-latex] (2) -- (3);
            \draw[->,-latex] (3) to[out=-35,in=-150] (2);
            \draw[->,-latex] (3) to[out=95,in=-160] (1);
		    \draw[->,-latex] (1) -- (3);
	\end{tikzpicture}
	\caption{Below peak}
    \end{subfigure}
    \begin{subfigure}[b]{0.28\linewidth}
	\centering
	\begin{tikzpicture}
		\tikzset{mynode/.style = {shape=circle,draw,inner sep=0pt,minimum size=25pt}} 
		\node[mynode] (1) at (1.5,1.5) {\footnotesize{Pref}};
		\node[mynode] (2) at (2.6,0) {\footnotesize{Del}};
		\node[mynode] (3) at (0.4,0) {\footnotesize{Add}};
            \draw[->,-latex] (2) -- (1);
            \draw[->,-latex] (2) -- (3); 
	\end{tikzpicture}
	\caption{At peak}
    \end{subfigure}
    \begin{subfigure}[b]{0.28\linewidth}
	\centering
	\begin{tikzpicture}
		\tikzset{mynode/.style = {shape=circle,draw,inner sep=0pt,minimum size=25pt}} 
		\node[mynode] (1) at (1.5,1.5) {\footnotesize{Pref}};
		\node[mynode] (2) at (2.6,0) {\footnotesize{Del}};
		\node[mynode] (3) at (0.4,0) {\footnotesize{Add}};
            \draw[->,-latex] (2) -- (1);
            \draw[->,-latex] (2) -- (3); 
	\end{tikzpicture}
	\caption{Above peak}
    \end{subfigure}
    \begin{subfigure}[b]{\linewidth}
	\centering
	\begin{tikzpicture}
		\node (1) at (0,0) {};
	\end{tikzpicture}
    \end{subfigure}
    \begin{subfigure}[b]{0.1\linewidth}
	\centering
	\begin{tikzpicture}
		\node[] (1) at (0,2) {{\FPDA{}:}};
            \node[] (1) at (0,0) { };
	\end{tikzpicture}
    \end{subfigure}
    \begin{subfigure}[b]{0.28\linewidth}
	\centering
	\begin{tikzpicture}
		\tikzset{mynode/.style = {shape=circle,draw,inner sep=0pt,minimum size=25pt}} 
		\node[mynode] (1) at (1.5,1.5) {\footnotesize{Pref}};
		\node[mynode] (2) at (2.6,0) {\footnotesize{Del}};
		\node[mynode] (3) at (0.4,0) {\footnotesize{Add}};
		    \draw[->,-latex] (1) to[out=-10,in=90] (2);
            \draw[->,-latex] (2) -- (1);
            \draw[->,-latex] (2) -- (3);
            \draw[->,-latex] (3) to[out=-35,in=-150] (2);
            \draw[->,-latex] (3) to[out=95,in=-160] (1);
		    \draw[->,-latex] (1) -- (3);
	\end{tikzpicture}
	\caption{Below peak}
    \end{subfigure}
    \begin{subfigure}[b]{0.28\linewidth}
	\centering
	\begin{tikzpicture}
		\tikzset{mynode/.style = {shape=circle,draw,inner sep=0pt,minimum size=25pt}} 
		\node[mynode] (1) at (1.5,1.5) {\footnotesize{Pref}};
		\node[mynode] (2) at (2.6,0) {\footnotesize{Del}};
		\node[mynode] (3) at (0.4,0) {\footnotesize{Add}};
		\draw[->,-latex] (1) to[out=-10,in=80] (2);
            \draw[->,-latex] (2) -- (1);
            \draw[->,-latex] (2) -- (3); 
		\draw[->,-latex] (1) -- (3);
	\end{tikzpicture}
	\caption{At peak}
    \end{subfigure}
    \begin{subfigure}[b]{0.28\linewidth}
	\centering
	\begin{tikzpicture}
		\tikzset{mynode/.style = {shape=circle,draw,inner sep=0pt,minimum size=25pt}} 
		\node[mynode] (1) at (1.5,1.5) {\footnotesize{Pref}};
		\node[mynode] (2) at (2.6,0) {\footnotesize{Del}};
		\node[mynode] (3) at (0.4,0) {\footnotesize{Add}};
            \draw[->,-latex] (2) -- (1);
            \draw[->,-latex] (2) -- (3); 
	\end{tikzpicture}
	\caption{Above peak}
    \end{subfigure}
    \vspace{-2mm}
\caption{ \small{ Manipulation trends for \emph{responsive preferences} under the \WPDA{} (top) and \FPDA{} (bottom) algorithm in the below peak/at peak/above peak regimes. An arrow from action $X$ to action $Y$ denotes the existence of an instance where $X$ is strictly more beneficial for the firm than $Y$. Each missing arrow from $X$ to $Y$ denotes that there is (provably) no instance where $X$ is more beneficial than $Y$.}}
\vspace{-2mm}
\label{fig:comparison}
\end{figure}

We now show how the usefulness of the three manipulation actions available to a firm depends on the relative values of its capacity and its peak. \Cref{fig:comparison} illustrates the comparison between the various manipulation actions under the \FPDA{} and \WPDA{} algorithms. Observe that in each of the three regimes in \Cref{fig:comparison}---\emph{below peak} (i.e., $c_f < p_f$), \emph{at peak} (i.e., $c_f = p_f$), and \emph{above peak} (i.e., $c_f > p_f$)---there exist scenarios where \texttt{\textup{Delete}} is strictly more beneficial than \texttt{\textup{Add}} (and similarly, more beneficial than \texttt{\textup{Pref}}). In fact, \texttt{\textup{Delete}} is the only manipulation that can be beneficial above peak. The \texttt{\textup{Add}} operation is only beneficial to a firm if its capacity is below the peak, irrespective of the matching algorithm. By contrast, at peak, \texttt{\textup{Pref}} is beneficial to a firm under \FPDA{} but is unhelpful under \WPDA{}.

In the rest of this section, we will present a representative set of proofs. We first discuss the comparison between \texttt{\textup{Delete}} and \texttt{\textup{Pref}} under the \WPDA{} algorithm. The comparison of these actions under the \FPDA{} algorithm is similar and we defer it to 
in \Cref{app:capvspref}.

\subsection{\texttt{\textup{Delete}} vs \texttt{\textup{Pref}} under \WPDA{}}

\subsubsection*{Below Peak} When the capacity of a firm is below its peak (i.e., $c_f < p_f$), there exists an instance where \texttt{\textup{Pref}} can outperform \texttt{\textup{Delete}} (as well as \texttt{\textup{Add}}) under the \WPDA{} algorithm.

\begin{example}[\texttt{\textup{Pref}} outperforms \texttt{\textup{Delete}} and \texttt{\textup{Add}} under \WPDA{} below peak]\label{ex:wosm-pref-below}
     Consider an instance $\I$ with three firms $f_1, f_2, f_3$ and four workers $w_1,w_2,w_3, w_4$. The firms have unit capacities (i.e., $c_1=c_2=c_3=1$) and have lexicographic preferences given by
    \begin{align*}
        w_1  &: f_2 \succ f_1 \succ f_3\succ  \emptyset & \hspace {0.15in} & f_1: w_4 \succ w_1 \succ w_2 \succ w_3 \\
        w_2 \, , w_3 &: f_1 \succ f_2 \succ f_3 \succ \emptyset & \hspace {0.15in} & f_2: w_3 \succ w_2 \succ w_1 \succ w_4 \\
        w_4 &: f_3 \succ f_1 \succ f_2 \succ \emptyset & \hspace {0.15in} & f_3: w_1 \succ w_4 \succ w_2 \succ w_3
    \end{align*}
    
    Under the \WPDA{} algorithm, firm $f_1$ is matched with $\{w_1\}$. 
    %
    %
    If $f_1$ uses \texttt{\textup{Add}} by switching to any capacity $c_1 \geq 2$, its \WPDA{} match is the set $\{w_2,w_3\}$. It is easy to verify that the peak for firm $f_1$ is $p_f(\I) = 2$. Thus, under $\I$, the capacity of firm $f_1$ is below the peak.
    %
    
    If $f_1$ uses \texttt{\textup{Pref}} in the instance $\I$ by misreporting its preferences to be $w_4 \succ w_2 \succ w_3 \succ w_1$, then its \WPDA{} match is $\{w_4\}$, which is more preferable for $f_1$ (according to its true preferences) than its match under \texttt{\textup{Add}}. On the other hand, using \texttt{\textup{Delete}} in the instance $\I$ (by reducing the capacity to $c_1=0$) is the worst outcome for $f_1$ as it is left unmatched.\qed
\end{example}

The next example shows that \texttt{\textup{Delete}} can outperform \texttt{\textup{Pref}} (which, in turn, outperforms \texttt{\textup{Add}}) under the \WPDA{} algorithm.

\begin{example}[\texttt{\textup{Delete}} outperforms \texttt{\textup{Pref}} and \texttt{\textup{Add}} under \WPDA{} below peak]\label{ex:wosm-del-below}
    Consider the following example with three firms $f_1, f_2, f_3$ and four workers $w_1,w_2,w_3,w_4$. The firms have capacities $c_1=2$ and $c_2=c_3=1$. The preferences are lexicographic and are given by
\begin{align*}
        w_1 &: f_3 \succ f_2 \succ f_1 \succ \emptyset & f_1: w_1 \succ w_2 \succ w_3 \succ w_4 \\
        w_2 &: f_1 \succ f_2 \succ f_3 \succ \emptyset & f_2: w_2 \succ w_4 \succ w_1 \succ w_3 \\
        w_3, \, w_4 &: f_1 \succ f_3 \succ f_2 \succ \emptyset & f_3: w_3 \succ w_4 \succ w_1 \succ w_2
    \end{align*}


    The worker-optimal stable matching in the above instance is 
    $$\mu_1= \{(\{w_2,w_3\},f_1), (w_1,f_2),(w_4,f_3)\}.$$

    The peak for firm $f_1$ is $3$; thus, we are in the below-peak setting.    
    
    It is easy to see that no matter what permutation of its preferences $f_1$ reports, it cannot be matched with a better set of workers than what it gets under $\mu_1$. In other words, \texttt{\textup{Pref}} is unhelpful.
    
    On the other hand, if $f_1$ decreaes its capacity to $c_1 = 1$, the \WPDA{} algorithm would result in the following matching:
    \[  \mu_2 =\{(w_1,f_1), (w_2,f_2),(w_3,f_3)\}   \]
    Due to lexicographic preferences, $f_1$ prefers $\mu_2$ to $\mu_1$, implying that \texttt{\textup{Delete}} outperforms \texttt{\textup{Pref}}. 
    Also note that if firm $f_1$ uses the \texttt{\textup{Add}} strategy without changing its preferences, it can only be matched with $\{w_2,w_3,w_4\}$, which, while a better outcome than \texttt{\textup{Pref}}, is worse than \texttt{\textup{Delete}}.
    \qed
\end{example}

\subsubsection*{At Peak} When the capacity of the firm is equal to the peak (i.e., $c_f = p_f$), \texttt{\textup{Pref}} becomes unhelpful under the \WPDA{} algorithm. This is because, in this case, the firm does not reject any worker and is therefore matched with the same set of workers as under truthful reporting~(\Cref{prop:Peak_Max_Proposals_WPDA}). 

To see why \texttt{\textup{Delete}} outperforms \texttt{\textup{Pref}} at peak, we refer the reader to \Cref{example:LP-firm-worse}. In the instance $\I'$ where firm $f_1$'s capacity is $2$ (i.e., at peak) and firm $f_2$'s capacity is $1$, firm $f_1$ is matched with $\{w_2,w_3\}$ in the worker-optimal stable matching. By decreasing its capacity to $1$, firm $f_1$ can be matched with $\{w_1\}$, which is a better outcome for it due to lexicographic preferences. The same example also shows that \texttt{\textup{Delete}} outperforms \texttt{\textup{Add}} because firm $f_1$'s outcome does not change after increasing capacity by any amount.

\subsubsection*{Above Peak} When the capacity of the firm is above its peak, \texttt{\textup{Pref}} again turns out to be unhelpful under the \WPDA{} algorithm. This is because a firm does not reject any proposals under \WPDA{} when its capacity is greater than or equal to its peak~(\Cref{prop:Peak_Max_Proposals_WPDA}). Thus, permuting the acceptable workers does not change the matched outcome of the firm. In \Cref{subsubsec:Appendix_DeletevsPref_WPDA}, we will prove a stronger result: \texttt{\textup{Pref}} is unhelpful in \emph{any} stable matching algorithm in the above-peak setting.

\begin{restatable}{theorem}{prefuseless}\label{thm:prefuseless}
    Under any stable matching algorithm, a firm cannot improve via preference manipulation~$(\texttt{\textup{Pref}})$ if its capacity is strictly greater than its peak.
\end{restatable}

By contrast, \texttt{\textup{Delete}} is helpful above peak under the \WPDA{} algorithm. To understand why, we refer back to \Cref{example:LP-firm-worse}. Let us consider a scenario where firm $f_1$'s capacity is $3$ and firm $f_2$'s capacity is $2$. In this case, the worker-optimal stable matching for this instance is $\mu_3$ in which firm $f_1$ is matched with $\{w_2,w_3\}$. If firm $f_1$ deletes two seats, it can get matched with $\{w_1\}$ under the worker-optimal stable matching. Due to lexicographic preferences, firm $f_1$ prefers $\{w_1\}$ over $\{w_2,w_3\}$. Thus,  \texttt{\textup{Delete}} outperforms \texttt{\textup{Pref}} (and, for the same reason, also \texttt{\textup{Add}}) above peak under the \WPDA{} algorithm.

The above example depends on the deletion of \emph{multiple} seats at once. If only one seat can be removed, then \texttt{\textup{Delete}} will be ineffective. This is because, in any above-peak instance, the worker-optimal matching assigns the same set of workers to the firm as it would in the at-peak instance. Therefore, reducing capacity by a unit amount does not change the set of workers matched with the firm.

\subsection{Manipulation via \texttt{\textup{Add}}}

We have seen above various instances where \texttt{\textup{Delete}} or \texttt{\textup{Pref}} are beneficial manipulation actions. We find that \texttt{\textup{Add}} is only useful when the firm's capacity is strictly below its peak. When a firm is at or above peak, increasing its capacity is not beneficial to it under any stable matching algorithm. This is because, in any above-peak instance, the set of workers matched with a firm in any stable matching is the same as that in the worker-optimal matching in the at-peak instance~(\Cref{lem:convergebeyondpeak}). Therefore, when the capacity is at least the peak, the firm cannot change its outcome \texttt{\textup{Add}} under the \WPDA{} algorithm. In the case of any other stable matching algorithm, the `original' outcome can only be weakly better than the worker-optimal stable matching at peak. Thus, increasing the capacity could only make the firm's outcome weakly worse.
%
%
\begin{corollary}\label{cor:Add-Useless-At-or-Above-Peak}
    When the capacity of a firm is greater than or equal to its peak, it cannot improve via adding capacity under any stable matching algorithm.
\end{corollary}
As a result, \texttt{\textup{Add}} can only be helpful when the firm's capacity is below the peak. We show a simple example where this is the case.

\begin{example}[\texttt{\textup{Add}} outperforms \texttt{\textup{Delete}} and \texttt{\textup{Pref}} below peak]\label{ex:masterlist}
%
Consider the following instance with two firms $f_1,f_2$ and five workers $w_1,w_2,w_3,w_4,w_5$. The workers have identical preferences, and so do the firms.
    \begin{align*}
        w_1,w_2,\dots,w_5 &: f_1 \succ f_2 \succ \emptyset & f_1, \, f_2 : w_1 \succ w_2 \succ w_3 \succ w_4\ \succ w_5
    \end{align*}
    The capacities of the firms are $c_1 = c_2 = 2$. Due to identical preferences, there is a unique stable matching given by
    \[\mu_1=\{(\{w_1,w_2\},f_1), (\{w_3,w_4\},f_2)\}.\]
    Note that the peak for firm $f_1$ is $5$ irrespective of $f_2$'s capacity. Thus, $f_1$ is in the below-peak regime.
    
    Since $f_1$ is already matched with its favorite set of two workers, it cannot improve through \texttt{\textup{Pref}}. Moreover, \texttt{\textup{Delete}} will worsen its outcome. However, $f_1$ can get a strictly better outcome $\{w_1,w_2,w_3\}$ by increasing its capacity. Thus, \texttt{\textup{Add}} outperforms both \texttt{\textup{Pref}} and \texttt{\textup{Delete}} under any stable matching algorithm. Furthermore, this observation holds under both strongly monotone and lexicographic extensions of the firms' preferences.
    %
    %
    \qed
    \end{example}

\subsubsection*{Strongly Monotone Preferences}
To conclude this section, we note that many of the examples discussed above have lexicographic preferences where a smaller set of workers may be preferred to a larger one. In \Cref{app:capvspref}, we study strongly monotone preferences where an agent always prefers a larger set of workers to a smaller one. Strongly monotone preferences have been studied in the fair division literature under the name of ``leveled preferences'' \citep{BNT21competitive,CC25fair,GHLT23unified}. For these preferences, we find an even starker difference in the usefulness of manipulation actions based on peak. We summarize these results in \cref{fig:smpcomparison}.

\section{Concluding Remarks}

We studied capacity modification in the many-to-one stable matching problem from qualitative, computational, and strategic perspectives and provided a comprehensive set of results. We first found that increasing the capacity of a firm by $1$ may be beneficial or harmful to the firm itself, but could never be harmful to any of the workers. We then computationally explored the effect of adding or deleting seats to firms on two objectives: matching a fixed pair under a stable matching or stabilizing a given matching. This capacity modification could be done either with only a global budget or with both global and firm-specific budgets on the number of seats to be added or removed.  We find that for matching a given pair of agents, while the problem could be tractably resolved with only global budgets, the presence of firm-specific budgets made the problem intractable. In contrast, with either type of budget, the problem of stabilizing a matching could be solved in polynomial time.

Our final set of results considered capacity modification from a strategic viewpoint and compared it with preference manipulation under the \WPDA{} and \FPDA{} algorithms. To this end, we introduced the concept of ``peak'', which is an important indicator of whether a particular manipulation action may be useful for an agent. We compared the relative usefulness of \texttt{Add}, \texttt{Delete}, and \texttt{Pref} under different regimes. On the way, we also provided structural results on capacity modification and the space of stable matchings. In particular, we found that when a firm's capacity is above peak, the set of workers matched with the firm in any stable matching is the same as its worker-optimal match when its capacity is equal to peak.


\paragraph{Future Work.} 
Our work highlights several potential directions for future research. Firstly, it would be interesting to explore algorithms for capacity modification when \emph{both} add and delete operations are allowed. At present, we only allow one action at a time. Secondly, one could consider a model with \emph{heterogeneous costs} where adding or removing a seat across different firms could incur different costs. Our work only considers homogeneous costs, in which each seat contributes equally to the budget. Thirdly, one could consider \emph{fairness} issues in the addition or removal of seats. Indeed, a firm may object to a disproportionate addition of seats to another firm or a disproportionate deletion of its own seats. Therefore, in addition to achieving stability objectives, it may be pertinent to add or remove seats equitably.


There is significant room for follow-up work on the \emph{manipulation} aspects as well. In particular, one could study situations where a firm can simultaneously misreport its preferences and change its capacity~\citep{KP09incentives}. Here, it would be interesting to explore the structure of optimal manipulation for a firm and whether such a manipulation can be computed in polynomial time. 

Our work focused primarily on lexicographic and strongly monotone preferences. For future work, one can consider other preference models, such as complementary preferences. Additionally, it would be interested to study \emph{constraints} on matched sets, such as matroid constraints. Conducting experiments on synthetic or real-world data to evaluate the frequency of availability of various manipulation actions is another natural direction to explore~\citep{RP99redesign}. 
%
%
Finally, it would be interesting to study capacity modification beyond the two-sided matching problem. Capacity constraints appear in various other contexts, including facility location \citep{ACLP20capacity,LWZ21budgeted,CFL+21mechanism} and fair division \citep{BB18fair,GHLT23unified,CN24repeatedly}. Studying the application of capacity modification in these areas is a promising direction for further research.

\section*{Acknowledgments}
We thank the anonymous reviewers for their helpful comments. We are also grateful to Manshu Khanna for pointing out the relevant work of \citet{ADV24capacity}. 
RV acknowledges
support from DST INSPIRE grant no. DST/INSPIRE/04/2020/000107, SERB grant no. CRG/2022/002621, and iHub Anubhuti IIITD Foundation. Part of this work was done as the SURA (Summer Undergraduate Research Award) project of SG and SS during May-July 2023. 

This material is based upon work supported by the National Science Foundation under Grant No. DMS-1928930 and by the Alfred P. Sloan Foundation under grant G-2021-16778, while SN was in residence at the Simons Laufer Mathematical Sciences Institute (formerly MSRI) in Berkeley, California, during the Fall 2023 semester. SN is currently supported by the NSF-CSIRO grant on ``Fair Sequential Collective Decision-Making."

\bibliographystyle{plainnat} 
\bibliography{ms}

\clearpage
\appendix
\begin{center}
   \LARGE{\textbf{Appendix}}
\end{center}

\section{Omitted Material from Section~\ref{sec:Preliminaries}}

\begin{restatable}[Illustration of Canonical Reduction]{example}{canonical}
Consider the following instance with two firms $f_1,f_2$, and three workers $w_1,w_2,w_3$. The workers' preferences are as follows:
    \begin{center}
        \begin{tabular}{c c}
          $w_1$: & $f_1\succ f_2 \succ \emptyset$\\
            $w_2$: & $f_2 \succ f_1 \succ \emptyset$\\
            $w_3$: & $f_2 \succ \emptyset$
        \end{tabular}
    \end{center}
The firms have strongly monotone preferences given by    
    \begin{align*}
        f_1 &: \{w_1,w_2\} \succ \{w_2\} \succ \{w_1\}\succ \emptyset\\
        f_2 &: \{w_1,w_2,w_3\} \succ \{w_2,w_3\}\succ \{w_1,w_3\}\succ 
        \{w_1,w_2\} \succ \{w_3\}\succ \{w_2\} \succ \{w_1\} \succ \emptyset
    \end{align*}

The capacities are $c_1 = c_2 = 2$. 
   
The canonical one-to-one instance consists of three women $q_1,q_2,q_3$ and four men $p^{1}_1,p^{2}_1,p^{1}_2,p^{2}_2$. Their preferences are given by
    \begin{center}
        \begin{tabular}{c c}
            $q_1$: & $p^{1}_1 \succ p^{2}_1 \succ p^{1}_2 \succ p^{2}_2 \succ \emptyset$\\
            $q_2$: & $p^{1}_2 \succ p^{2}_2 \succ p^{1}_1 \succ p^{2}_1 \succ \emptyset$\\
            $q_3$: & $p^{1}_2 \succ p^{2}_2 \succ \emptyset$
        \end{tabular}
    \end{center}

    \begin{align*}
        p^{1}_1 &: q_2 \succ q_1 \succ \emptyset\\
        p^{2}_1 &: q_2 \succ q_1 \succ \emptyset\\
        p^{1}_2 &: q_3 \succ q_2 \succ q_1 \succ \emptyset\\
        p^{2}_2 &: q_3 \succ q_2 \succ q_1 \succ \emptyset
    \end{align*}

The unique stable matching in the one-to-one instance is $$\mu' = \{(q_1, p^{1}_1), (q_2,p^{2}_2),(q_3,p^{1}_2)\}$$ which maps to $$\mu = \{(w_1, f_1), (\{w_2,w_3\}, f_2)\}$$ in the many-to-one instance.
\label{example:canonical}
\end{restatable}

\section{Omitted Material from Section~\ref{sec:Capacity_Modification_Trends}}
\label{subsec:Capacity_Trends_Appendix}

We now give a proof of \Cref{prop:nofirmworse}.
\nofirmworse*

\Cref{prop:nofirmworse} was established in the works of~\citet{KU06games} and~\citet{kojima2007can}. Below, we recall the proof with the goal of highlighting which part of the argument requires (or does not require) strong monotonicity.

To prove \Cref{prop:nofirmworse}, we will first show that upon capacity increase by a firm $f$ under the \WPDA{} algorithm, if the number of workers matched with $f$ does not change, then firm $f$ must be matched with the same \emph{set} of workers~(\Cref{lemma:firm-increase}). This observation does not rely on the preferences being strongly monotone.

\begin{restatable}{lemma}{}
Let $\mu$ and $\mu'$ denote the worker-optimal stable matching before and after a firm $f$ increases its capacity by $1$. If the number of workers matched with $f$ remains unchanged, i.e., if $|\mu'(f)| = |\mu(f)|$, then the set of workers matched with $f$ also remains unchanged, i.e., $\mu'(f) = \mu(f)$.
\label{lemma:firm-increase}
\end{restatable}
\begin{proof}
Suppose, for contradiction, that $\mu'(f) \neq \mu(f)$.

Let $\I$ and $\I'$ denote the instances before and after capacity increase by firm $f$. Consider the execution of the \WPDA{} algorithm on these instances. Let $r$ be the \emph{first} round in which the two executions differ. Then, the two executions must differ in the \emph{rejection} phase of round $r$ such that an acceptable worker, who was rejected by firm $f$ under $\I$, is tentatively accepted under $\I'$. This implies that under the old instance $\I$, firm $f$ had strictly more than $c_f$ tentative proposals in hand after the proposal phase of round $r$. After round $r$ under $\I'$, firm $f$ will be tentatively matched with $c_f + 1$ workers (in other words, the firm becomes saturated in round $r$). Since the number of workers matched with any fixed firm weakly increases during the execution of the worker-proposing algorithm, it follows that $|\mu'(f)| = c_f + 1$. However, this contradicts the fact that $|\mu'(f)| = |\mu(f)|$ since $|\mu(f)| \leq c_f$. Hence, it must be that $\mu'(f) = \mu(f)$.
\end{proof}

A similar argument as in the proof of \Cref{lemma:firm-increase} shows that the number of workers matched with a firm under the \WPDA{} algorithm cannot decrease if its capacity increases. Again, strong monotonicity is not required for this claim.

\begin{restatable}{proposition}{}
Let $\mu$ and $\mu'$ denote the worker-optimal stable matching before and after a firm $f$ increases its capacity by $1$. Then, $|\mu'(f)| \geq |\mu(f)|$.
\label{prop:firm-increase}
\end{restatable}

We will now use strong monotonicity to finish the proof of \Cref{prop:nofirmworse}.
\begin{proof} (of \Cref{prop:nofirmworse})
From \Cref{prop:firm-increase}, we know that $|\mu'(f)| \geq |\mu(f)|$. If $|\mu'(f)| = |\mu(f)|$, then from \Cref{lemma:firm-increase}, we know that $\mu'(f) = \mu(f)$ and the claim follows. Otherwise, it must be that $|\mu'(f)| > |\mu(f)|$. Then, due to strongly monotone preferences, it follows that $\mu'(f) \>_f \mu(f)$, as desired.
\end{proof}

\section{Omitted Material from Section~\ref{sec:Computational_Results}}
\label{sec:Appendix_Computational_Results}

This section provides the proofs of results listed in \Cref{tab:manipulation} whose discussion was omitted from \Cref{sec:Computational_Results}. 

\subsection{Deleting Capacity To Match A Pair: Unbudgeted}

We now look at the problem of deleting capacity to match a worker-firm pair $(f^{*},w^{*})$ in the unbudgeted setting. Formally the problem is defined as follows:

\problembox{\UnrestrictedDeleteCapacityMatchPair{}}{An instance $\I = \langle F,W,C, \> \rangle$, a worker-firm pair $(w^*,f^*)$, a global budget $\ell \in \whole$}{Does there exist a capacity vector $\overline{C} \in (\whole)^n$ such that $C \geq \overline{C}$, $|C - \overline{C}|_{1} \leq \ell$, and $w^*$ and $f^*$ are matched in some stable matching for the instance $\I' = \langle F,W,\overline{C}, \> \rangle$?}

We can solve this problem in polynomial time, as we will show now.

\begin{restatable}{theorem}{DelCapPairPoly}\label{theorem:Delete-Capacity-Pair-Polynomial}
\UnrestrictedDeleteCapacityMatchPair{} can be solved in polynomial time.
\end{restatable}

We will prove this theorem by via the canonical reduction to an equivalent problem for the one-to-one setting. First, we state the following problem in which we try to match a man-woman pair by deleting only men.

\problembox{\ConstructiveExistsDeleteMen{}}{An instance $\I = \langle P,Q,\> \rangle$, a man-woman pair $(p^*,q^*)$ and a budget $\ell \in \whole$.}{Does there exist a set $\overline{P} \subseteq P$ such that $|\overline{P}| \leq \ell$ and the pair $(p^*, q^*)$ is part of at least one matching in $\langle P \setminus \overline{P}, Q, \> \rangle$?}

Note that this problem is not a special case, but rather a variant of \ConstructiveExistsDelete{} defined in \citep{BBH+21bribery} (originally, this was called \textup{\textsc{Constructive-Exists-Delete}}). \ConstructiveExistsDelete{} effectively asks if we can delete at most $\ell$ men or women in order to match $(p^*,q^*)$ under a stable matching. Observe that for a yes instance of  \ConstructiveExistsDeleteMen{} implies a yes instance for the same budget of $\ell$ for \ConstructiveExistsDelete{}, but not vice versa. 
Consequently, \ConstructiveExistsDeleteMen{} cannot be solved by the tools discussed in \citep{BBH+21bribery}. We now modify the polynomial time \ConstructiveExistsDelete{} algorithm decribed in \cite[Theorem 2]{BBH+21bribery} to solve \ConstructiveExistsDeleteMen{} in polynomial time. 

In  order for the given $(p^*,w^*)$ pair to be stably matched, we need to eliminate all blocking pairs which would be formed by $p^*$ or $q^*$. For the \ConstructiveExistsDeleteMen{} problem, our only tool to do this is to either match all such distracting agents to other agents they prefer to $p^*$ or $q^*$ or delete the distracting men.  
To find the set of agents to be deleted, we construct a instance where all the men and women preferred by $p^*$ and $q^*$ to each other, can only be matched with someone they prefer over $q^*$ or $p^*$ or remain unmatched. All the unmatched agents are the agents which necessarily must be deleted. Hence, if this set contains a woman or is larger than the given budget, then we cannot achieve our goal. Otherwise, we claim that this set of agents, which consists of only men is necessary and also sufficient to match $(p^*, q^*)$ in some matching.



\begin{algorithm}[!ht]
  \DontPrintSemicolon
  \KwIn{A \ConstructiveExistsDeleteMen{} instance $ \I = \langle P, Q, \>  , \ell, p^*, q^* \rangle$ }
  \KwOut{A set $\overline{P}\subset P$ to be removed and a stable matching $\mu$ in $\I \setminus \overline{P}$ containing $(p^*, q^*)$}
  $A \gets $ the set of men which $q^*$ prefers over $p*$\;
   $B \gets $ the set of women which $p^*$ prefers over $q^*$\;
  Initialize $\I' = \langle P', Q', \succ'\rangle$ where $P' = P \setminus \{p^*\}$, $Q' = Q \setminus \{q^*\}$ and preferences $\succ' \gets \succ$\\
  For all $p\in A$ delete acceptability of all women $q$ s.t. $q^*\succ_p q$ under $\succ'_p$\\
  For all $q\in B$ delete acceptability of all men $p$ s.t. $p^*\succ_q p$ under $\succ'_q$\\
  Compute a stable matching $\mu'$ in $\I'$ \;

  Let $\overline{P} \gets \{p\in A|p $ is unmatched under$\mu'\}$ \;
  Let $\overline{Q} \gets \{q\in B|q $ is unmatched under$\mu'\}$ \;

  \eIf{$|\overline{P}| \leq \ell$ and $\overline{Q}=\emptyset$}
  {
     $\mu \gets \mu' \cup (p^*, q^*)$\;
  }
  {
    $\overline{P}\gets \emptyset$ and $\mu \gets \emptyset$ \Comment*[r]{Cannot match $(p^*, q^*)$ within budget}
  }    
  \textbf{Return} $\overline{P}$ and $\mu$
\caption{\ConstructiveExistsDeleteMen{} algorithm}

\label{CEDM}
\end{algorithm}

To prove that \Cref{CEDM} correctly solves \ConstructiveExistsDeleteMen{}, we need to use the following lemmas on the influence of a delete operation on the set of agents matched in a stable matching.

\begin{restatable}[\citet{BBH+21bribery}]{lemma}{DeleteLemma}
Let $\I = \langle P, Q, \succ \rangle$ be a stable matching instance and $a\in P \cup Q$ be some agent. Then, there exists at most one agent $a' \in P\cup Q$ which was unassigned in $\I$, and is matched in $\I \setminus \{a\}$.
\label{lemma:DeleteLemma}
\end{restatable}

\begin{restatable}{lemma}{DeleteAgent}
Let $\I = \langle P, Q, \succ \rangle$ be a stable matching instance and $p \in P$ be some man. Then, there cannot exist a woman $q \in Q$ who was unassigned in a stable matching $\I$, and is matched in a stable matching in $\I \setminus \{p\}$. 
\label{lemma:DeleteAgent}
\end{restatable}

\begin{proof}
    Let $\mu$ and $\mu'$ be the women-optimal stable matching in $\I \setminus \{p\}$ and $\I$ respectively. Using \Cref{prop:no-women-worse-wosm} we know that $\mu'$ is weakly preferred by all women over $\mu$.
    
    For the sake of contradiction, assume that there is a woman $q \in Q$ such that $q$ is unassigned in some stable matching of $\I$, and is matched in some stable matching of $\I \setminus \{p\}$. \Cref{prop:RuralHospitals} states that the set of unmatched agents in any stable matching in a given instance is the same. Hence, $q$ is also unmatched in $\mu'$. Since, $\mu'$ is weakly preferred by all women over $\mu$, $q$ is also unmatched in $\mu$. Again using  \Cref{prop:RuralHospitals}, we get that $q$ is unmatched in all stable matchings of $\I \setminus \{p\}$, which is a contradiction.
\end{proof}

We now proceed with the proof of correctness of \Cref{CEDM}. Our proof differs from the proof given by \citet{BBH+21bribery} for their algorithm in the case when $\overline{P}$ contains women. We will call the set of men and women in the sets $A$ and $B$ (as defined in \Cref{CEDM}) as distracting men and women, respectively. 

\begin{restatable}{lemma}{Constructive-Exists-Delete-Men-Polynomial} \ConstructiveExistsDeleteMen{} can be solved in polynomial time by \Cref{CEDM}.
\label{lemma:Constructive-Exists-Delete-Men-Polynomial}
\end{restatable}

\begin{proof} 
We will first show that \Cref{CEDM} solves any given instance $\langle \I = \langle P, Q, \>\rangle , \ell, p^*, q^* \rangle$ of \ConstructiveExistsDeleteMen{} correctly. Consider the set $\overline{P}$ constructed by the algorithm.

\paragraph{Existence of a solution.} Firstly, we will show that if $|\overline{P}| \leq \ell$ and $\overline{Q}$ is empty, then the given instance is solved by $\overline{P}$ and $\mu = \mu' \cup (p^*, q^*)$ is a stable matching in the instance $\I \setminus \overline{P}$. Note that $\overline{P}$ is within budget, hence it is a valid set of men that can be deleted. For the sake of contradiction, assume that there exists a blocking pair $(p, q)$ in $\mu$. 

It must be that $p \neq p^*$. This follows from the fact that all agents that $p^*$ prefers over $q^*$ are either deleted or matched to someone they prefer over $p^*$. Similarly, it must be that $q \neq q^*$. As $(p, q)$ is not a blocking pair in $\I' \setminus \overline{P}$, it contains an agent $a \in (A \cup B)\setminus \overline{P}$, and $a$ prefers $q^*$ to $q$ if $a = p$ or $a$ prefers $p^*$ to $p$ if $a = q$. Without loss of generality, we assume that $a = p$. As $p$ is matched in $\mu'$, he prefers $\mu'(p) = \mu(p)$ to $q^*$. Thus, $p$ prefers $\mu(p)$ to $q$, a contradiction to $(p,q)$ being blocking for $\mu$. Hence, there are no blocking pairs in $\mu$ and $\overline{P}$ solves the given instance.

\paragraph{No solution within budget.} Now we will show that if $|\overline{P}| > \ell$ (while $\overline{Q}$ is empty), then it is not possible to solve the given instance without exceeding the budget. For the sake of contradiction, assume that there exists a set of men $P^{\dagger} = \{p_1, \dots, p_k\}$ with $k \leq \ell$ such that $\I \setminus P^{\dagger}$ admits a stable matching $\mu$ containing $(p^*, q^*)$. For each $i \in \{0, 1, \dots, k\}$, let $\mu'_i$ be a stable matching in $\I'\setminus \{p_1, \dots, p_i\}$. 

From the rural hospitals theorem (\Cref{prop:RuralHospitals}), we have that for a given instance, the set of unmatched agents in any stable matching is the same. Hence, the exact choice of stable matching for $\mu'_i$ does not matter for the purposes of ensuring that there are no unmatched distracting men. By the definition of $\overline{P}$, all agents in $\overline{P}$ are unmatched under $\mu_0'$. Also, each distracting agent $a \in A\cup B$ is either part of $P^{\dagger}$ or prefers $\mu(a)$ to $p^*$ or $q^*$ due to the stability of $\mu$. Hence, $\mu$ restricted to $\I' \setminus P^{\dagger}$ is also stable. 

In particular, every (unmatched distracting) man $p \in \overline{P}$ is either contained in $P^{\dagger}$ or matched in $\mu'_k$. Recall that the size of $P^{\dagger}$ is such that  $|P^{\dagger}|=k \leq \ell < |\overline{P}|$. Consequently, there exists an $i$ and two agents $p, p' \in \overline{P}$ which are unmatched in $\mu'_{i-1}$ and not contained in $\{p_1, \dots, p_{i-1}\}$ but matched in $\mu'_{i}$ or contained in $\{p_1, \dots, p_i\}$. By \Cref{lemma:DeleteLemma}, it is not possible that both $p$ and $p'$ are unmatched in $\mu'_{i-1}$ but matched in $\mu'_{i}$. 

The only remaining case is that one was previously unmatched and is now matched, while the other was removed. To this end, without loss of generality, let $p = p_i$ with $p$ being unmatched in $\mu'_{i-1}$, and $p'$ is matched in $\mu'_{i}$ but unmatched in $\mu'_{i-1}$. However, as $p$ is unmatched in $\mu'_{i-1}$, deleting it does not change the set of matched agents. This contradicts $p'$ being matched in $\mu'_{i}$ but unmatched in $\mu'_{i-1}$. Hence, such a set of men $P^\dagger$ cannot exist.

\paragraph{No solution due to distracting women.} Finally, we will show that if $\overline{Q}$ is non-empty, then it is not possible to solve the given instance. Recall that $\overline{Q}$ is the set of women that $p^*$ prefers to $q^*$ who are unmatched under any stable matching in the instance $\I'$, constructed in \Cref{CEDM}. 

Fix any set of men $P^{\dagger}$ not containing $p^*$. From \Cref{lemma:DeleteAgent}, we know that deleting a man from $I'$ and cannot match a previously unmatched woman under a stable matching. Consequently, all women in $\overline{Q}$ will continue to remain unmatched  under any stable matching under $\I'\setminus P^{\dagger}$. 

For contradiction, let $\I \setminus P^{\dagger}$ admit a stable matching $\mu$ containing $(p^*, q^*)$ when $\overline{Q}\neq \emptyset$. Note that, as $\mu$ is stable, for each $q\in \overline{Q}$, it must be that $\mu(q)\succ_q p^*$. As a result, each woman $q\in \overline{Q}$ must be matched under $\mu$, and  $\mu \setminus (p^*,q^*)$  must be stable under $\I'\setminus P^{\dagger}$. This is a contradiction. Thus, there cannot exist a $P^{\dagger}$ such that $\I \setminus P^{\dagger}$ admits a stable matching $\mu$ containing $(p^*, q^*)$.




This completes the proof of correctness. The complexity clearly follows from the fact that we can create the new instance in linear time and a stable matching can be computed in polynomial time.
\end{proof}

In order to solve \UnrestrictedDeleteCapacityMatchPair{}, via the canonical reduction to an equivalent one-to-one instance, we require some more machinery. The canonical reduction will construct $c_f$ copies of each firm $f\in F$, and we would be happy to have the worker matched to any of those copies. To this end, we now define a new one-to-one matching problem in which our goal is to match a woman $q^*$ with any man from a given set $P^*$. We will show that this problem can be solved in polynomial time too via a reduction to multiple instances of \ConstructiveExistsDeleteMen{}.

\problembox{\MultipleConstructiveExistsDeleteMen{}}{An instance $\I = \langle P,Q,\> \rangle$, a subset $P^{*} \subseteq P$ of men, a woman $q^*$, and a budget $\ell \in \whole$.}{For at least one man $p^* \in P^*$, does there exist a set $\overline{P} \subseteq P$ such that $|\overline{P}| \leq \ell$ and the pair $(p^*, q^*)$ is part of at least one matching in $\langle P \setminus \overline{P}, Q, \> \rangle$?}

The idea is that, for every man $p$ in $P^*$, we will try to match him with $q^*$ by creating an instance of \ConstructiveExistsDeleteMen{} in which our goal is to match $(p, q^*)$.

\begin{restatable}{lemma}{Delete-Multiple-Men-to-Delete-Men-Reduction}
\MultipleConstructiveExistsDeleteMen{} can be solved in polynomial time.
\label{lemma:Delete-Multiple-Men-to-Delete-Men-Reduction}
\end{restatable}

\begin{proof}
We will show that given an instance of \MultipleConstructiveExistsDeleteMen{}, we can resolve it by considering if one of $|P^*|$ instances of \ConstructiveExistsDeleteMen{} admits a solution.

Given an instance of \MultipleConstructiveExistsDeleteMen{} with $ \I = \langle P,Q,\>\rangle,$ along with a budget $ \ell$, the set of men $ P^{*}$ and the woman$ q^{*}$. For every man $p \in P^{*}$, we create an instance of \ConstructiveExistsDeleteMen{} where the aim is to match $p$ and $q^*$  for the same budget $\ell$ and instance $\langle P,Q, \succ \rangle$. Let's call this instance $\I_p$. We shall now show that $\I$ admits a solution if and only if for some $p\in P^*$, $\I_p$ admits a solution. 

First assume that $\I$ admits a solution where by removing $P^{\dagger}$, there exists a stable matching $\mu$ under $\I\setminus P^{\dagger}$ where $\mu(q^*)\in P^*$. Fix one such $\mu$ and let $\mu(q^*)=p^*\in P^*$. 

Clearly, $p^*\notin P^{\dagger}$ and $|P^{\dagger}|\leq \ell$. Consequently, it is straightforward to see that $\mu$ is also stable under $\I_{p^*}\setminus P^{\dagger}$, and thus $\I_{p^*}$ admits a solution. 

The converse follows analogously. If  for some $p\in P^*$  a  $\mu$ containing $(p^*,q^*)$ is stable under $\I_p \setminus P^{\dagger}$, then as the set of agents, the preferences and budgets are identical to $\I$, $\mu$ must be stable under $\I\setminus P^{\dagger}$.

Recall that by \Cref{lemma:Constructive-Exists-Delete-Men-Polynomial}, \ConstructiveExistsDeleteMen{} is solvable in polynomial time. Further, as the size of $P^*$ is at most the total number of men in $P$, we have that $|P^*|$ instances of \ConstructiveExistsDeleteMen{} can be solved in polynomial time. As a result,  \MultipleConstructiveExistsDeleteMen{} can be solved in polynomial time.
\end{proof}

Finally, via the canonical one-to-one reduction, we can prove the following theorem.


\DelCapPairPoly*

\begin{proof} 
We give a reduction from \UnrestrictedDeleteCapacityMatchPair{} to \MultipleConstructiveExistsDeleteMen{}. Let the  given an instance of \UnrestrictedDeleteCapacityMatchPair{} be $\I = \langle F,W,C, \>, \ell, f^{*}, w^{*} \rangle$. We now construct an instance of \MultipleConstructiveExistsDeleteMen{} $\I' = \langle P',Q', \succ' , \ell, P^{*}, q^{*} \rangle$  in which $\I'$ is the canonical one-to-one reduction of $\I$,  $q^*$ is the woman corresponding to $w^{*}$, and $P^{*}$ is the set of men corresponding to copies of $f^{*}$.

    First assume that the given instance of \UnrestrictedDeleteCapacityMatchPair{} is solved by $\overline{C}$ such that $\langle F, W, \overline{C}, \> \rangle$ admits a stable matching $\mu$ containing $(f^*, w^*)$. We shall now define a set $P^{\dagger}$ of men to delete. 
    
    For each firm $f$ s.t. $c_f>\overline{c}_f$, add the $|c_f-\overline{c}_f|$ least preferred men corresponding to $f$ to $P^{\dagger}$. Then, deleting $P^{\dagger}$ from $\I'$ solves the reduced instance of \MultipleConstructiveExistsDeleteMen{}. Observe that $\I' \setminus P^{\dagger}$ is the canonical one-to-one reduction of $\langle F, W, \overline{C}, \> \rangle$. Hence, the stable matching corresponding to $\mu$ in $\I' \setminus P^{\dagger}$ will match $q^*$ with some man corresponding to a copy of $f^{*}$, and this man will be in $P^{*}$.
    
    We can analogously construct a solution of the given instance of \UnrestrictedDeleteCapacityMatchPair{} given a solution of the reduced instance of \MultipleConstructiveExistsDeleteMen{}. 
    
    Consequently, from \Cref{lemma:Constructive-Exists-Delete-Men-Polynomial}, \UnrestrictedDeleteCapacityMatchPair{} can be solved in polynomial time. 
\end{proof}


\subsection{Deleting Capacity To Match A Pair: Budgeted}
We now look at the budgeted version of the previous problem, which can be defined formally as follows:

\problembox{\RestrictedDeleteCapacityMatchPair{}}{An instance $\I = \langle F,W,C, \> \rangle$, a firm-worker pair $(f^*,w^*)$, a global budget $\ell \in \whole$, and an individual budget $\ell_f \in \whole$ for each firm $f \in F$.}{Does there exist a capacity vector $\overline{C} \in (\whole)^n$ such that $C \geq \overline{C}$, $|C - \overline{C}|_{1} \leq \ell$ and for every $f \in F$, $|c_f - \overline{c}_f| \leq \ell_f$, and $w^*$ and $f^*$ are matched in some stable matching for the instance $\I' = \langle F, W,\overline{C}, \> \rangle$?}

In contrast to the unbudgeted setting, this problem is \NPH{}. We can show this from a straightforward reduction from \ConstructiveExistsAddMen{}, which was shown by \cite{BBH+21bribery} to be \NPH{} (\Cref{prop:Constructive-Exists-Add-Men-W1-Hard}). The reduction creates a worker for each man and a firm with capacity $1$ for each woman. Additionally, for each man in $P_{add}$, we create another auxiliary firm with capacity $1$, which would be the corresponding man's first preference. Deleting the capacity of this auxiliary firm will correspond to adding the corresponding man from $P_{add}$.

\begin{restatable}{theorem}{DelCapHard}\RestrictedDeleteCapacityMatchPair{} is \NPH{}. 
\label{theorem:Delete-Capacity-W1-Hard}
\end{restatable}


\begin{proof}
Given an instance of \ConstructiveExistsAddMen{} $\langle P_{orig},Q,\>, \ell, P_{add}, p^{*}, q^{*} \rangle$, we create a new instance of \RestrictedDeleteCapacityMatchPair{}  as follows:

For every woman $q \in Q$, create a firm $f_{q}$ with capacity 1 and individual budget $\ell_{f_q}=0$. That is, no capacity can be removed from these firms. For every man $p \in P_{orig}\cup P_{add}$, create a worker $w_{p}$. The preference relation of these agents over each other is kept the same as in the original instance. 

Additionally, for every addable man $p \in P_{add}$, create a firm $f_{p}$ with capacity 1 and individual budget $\ell_{f_p}=1$. These firms only consider their corresponding worker $w_{p}$ acceptable. Similarly, these workers also prefer their corresponding firm over other firms corresponding to women. In this way, these workers and firms corresponding to each man in $P_{add}$ are matched in every stable matching the constructed instance. Finally, we set the global budget as $\ell$, and the firm worker pair to be matched as $f^{*} = f_{q^{*}}$ and $w^{*} = w_{p^{*}}$.

We first consider the case where \ConstructiveExistsAddMen{} has a solution. Let the given instance be solved by adding $\overline{P}\subseteq P_{add}$. Let the corresponding stable matching containing $(p^*, q^*)$ be $\mu$. We will show that we can construct a solution for the constructed instance by reducing the capacity of each firm $f_{p}$ corresponding to a man in $p\in \overline{P}$.

Formally, the required capacity vector $\overline{C}$ is $\overline{c}_{f_q}=c_{f_q}=1$ and $\overline{c}_{f_p}=0$ if $p\in \overline{P}$ otherwise $\overline{c}_{f_p}=1$.
Observe that this capacity vector satisfies the individual budget constraints on the firms. Define a matching $\mu'$ in the instance $\langle F,W, \overline{C}, \> \rangle$ as follows: 

\[\mu' = \{(f'_{\mu(p)},w'_{p}) : p \in P_{orig} \cup \overline{P}, p \textnormal{ is matched in } \mu\} \cup \{( f'_{p},w'_{p}) : p \in P_{add}\setminus \overline{P}\}\]
    
Clearly, $\mu'$ contains the pair $(f^{*},w^{*})$ and satisfies the capacity vector we defined earlier. We can now prove that $\mu'$ is stable. 

None of the workers or firms corresponding to a man in $P_{add}\setminus \overline{P}$ can form a blocking pair because they are matched with their most preferred firm. Assume for the sake of contradiction that $(f'_{q},w'_{p})$ is a blocking pair in $\mu'$ with $p \in P_{orig} \cup \overline{P}$. This would imply that $(p,q)$ is a blocking pair in $\mu$ because the preferences over corresponding agents are kept same in $\mu'$ and $\mu$. This is contradicts the stability of $\mu$. As a result, $\mu'$ must be stable. Hence, the constructed instance $\I'$ of \RestrictedDeleteCapacityMatchPair{} is solved by $\overline{C}$.

Conversely, let the constructed instance $\I'$ have a solution. Let $\overline{C}$ be such that the modified instance contains a stable matching $\mu'$ with the pair $(f^{*},w^{*})$. We now show that adding the men corresponding to the firms whose capacity is deleted solves the given instance $\I$. 

Recall that the individual budgets are set such that we can only delete capacity from the firms corresponding to a man in $P_{add}$. Formally, let this set of men (whose corresponding firms had their capacity deleted)  be $\overline{P} = \{p \in P_{add} : \overline{c}_{f'_{p}} = 0\}$. Note that the size of this set $|\overline{P}|\leq \ell'= \ell$. Define a matching $\mu$ for $\I \cup \overline{P}$ as follows:

$$\mu = \{(p, q) : \mu'(w'_{p}) = f'_{q}, p\in {P_{orig} \cup \overline{P}}\}$$

Note that $\mu$ must contain the pair $(p^{*}, q^{*})$. Assume, for contradiction, that there is a blocking pair $(p, q)$ for $\mu$. This would imply $( f'_{q},w'_{p})$ is a blocking pair for $\mu'$ because the preferences over corresponding agents identical across the instances. This contradicts the stability of $\mu'$. 
As a result, $\mu$ must be stable. Hence, $\overline{P}$ solves the given instance $\I$.

Thus, \RestrictedDeleteCapacityMatchPair{} is \NPH{}.
\end{proof}

\subsection{Adding Capacity To Stabilize A Matching}

We now shift our focus to a different problem. In this problem, our goal is to increase the capacity of firms in order to make a subset of the given matching stable. In the budgeted version, we have a global budget on the total number of seats added to the instance along with individual budgets on the number of seats that can be added to each firm. We assume that all preferences are complete. That is, for each firm, all workers are acceptable and vice versa. The problem can be defined as:

\problembox{\RestrictedAddCapacityStabilizeMatching{}}{
An instance $\I = \langle F, W, C, \> \rangle$, individual budget $\ell_f \in \whole$ for each firm $f \in F$, a global budget $\ell \in \whole$, a matching $\mu^*$ with $|\mu^*(f)| \leq c_f + \ell_f$ for each $f \in F$. Preferences $\>$ are complete.
}{
Does there exist a capacity vector $\overline{C} \in (\whole)^n$ such that $\overline{C} \geq C, |\overline{C} - C|_1 \leq \ell$ and $|\overline{c}_f - c_f| \leq \ell_f$ for each $f \in F$ and there is a matching $\mu$ stable for $\I' = \langle F, W, \overline{C}, \> \rangle$ which is a subset of $\mu^*$, i.e., $\mu(f) \subseteq \mu^*(f)$ for each $f$ and $\mu(w) = \mu^*(w)$ or $\emptyset$ for each $w$ ?
}

We will show that \RestrictedAddCapacityStabilizeMatching{} can be solved in polynomial time by using the fact that its canonical one-to-one extension can be solved in polynomial time, as shown by \cite{BBH+21bribery}.

\begin{restatable}{theorem}{BudgetAddStabilize}
    \RestrictedAddCapacityStabilizeMatching{} can be solved in polynomial time.
    \label{theorem:RestrictedAddCapacityStabilizeMatching-Polynomial}
\end{restatable}

We define \ExactExistsAddMen{}, which is a special case of the problem \ExactExistsAdd{} defined by \cite{BBH+21bribery}. \ExactExistsAdd{}, which looks at add agents from a given set of ``addable'' agents to make a given matching stable. \ExactExistsAddMen{} is a special case where all the addable agent are men. 

\problembox{\ExactExistsAddMen{}}{
An instance $\I = \langle P_{orig}, Q, \> \rangle$, a set of addable men $P_{add}$, a global budget $\ell$, and a matching $\mu^*$ defined over $P_{orig} \cup P_{add} \cup Q$. Additionally, the preferences $\>$ are complete and defined over $P_{orig} \cup P_{add} \cup Q$.
}{
Does there exist a set of men $\overline{P} \subseteq P_{add}$ such that $|\overline{P}| \leq \ell$ and the matching $\mu^*$ when restricted to $P_{orig} \cup \overline{P} \cup Q$ is stable. More formally, does there exist a stable matching $\mu$ in $\I \cup \overline{P}$ such that for each $p \in P_{orig} \cup \overline{P}$, $\mu(p) = \mu^*(p)$ and for each woman $q \in Q$, $\mu(q) = \mu^{*}(q)$ if $ \mu^{*}(q) \in P_{orig} \cup \overline{P}$ or otherwise $\mu(q) = \emptyset$?
}

Since \ExactExistsAdd{} was shown to be solvable in polynomial time by \cite{BBH+21bribery}, we can solve \ExactExistsAddMen{} too. For completeness, we provide the algorithm for this, based on \cite[Algorithm 1]{BBH+21bribery} in \Cref{alg:EEAM}.

\begin{lemma}[\citet{BBH+21bribery}]
    \ExactExistsAddMen{} can be solved in polynomial time.
    \label{lemma:ExactExistsAddMen-Polynomial}
\end{lemma}

\begin{algorithm}[!ht]
    \DontPrintSemicolon
  \KwIn{An instance $ \I = \langle P_{orig}, Q, \> , P_{add}, \ell,\mu^{*} \rangle$ }
  \KwOut{A set $\overline{P}\subseteq P_{add}$ of men and a matching $\mu \subseteq \mu^*$.}
   Initialize $\overline{P} \gets \emptyset$. \;
   Let $\mu\gets \{(p,q)\in \mu^*|p\in P_{add}\}$\;
   Initialize $B \gets \{q\in Q| $q is involved in a blocking pair under $\mu \}$\;

  \For{each $q\in B$}{
    \If{$\mu^*(q)\in P_{add}$}{
        $\overline{P}\gets \overline{P}\cup \mu^*(q)$\;
    }
  }
  Update $\mu\gets \{(p,q)\in \mu^*|p\in P_{add}\cup \overline{P}\}$\;
   \If{$\mu$ is not stable OR $|\overline{P}|> \ell$}
   {
        Update $\overline{P}\gets \emptyset$, $\mu \gets \emptyset$ \Comment*[r]{Cannot stabilize $\mu^*$ within the budget.}
   }
    \textbf{Return} $\overline{P}$ and $\mu$
  
\caption{\ExactExistsAddMen{} algorithm}

\label{alg:EEAM}
\end{algorithm}

We now show that \RestrictedAddCapacityStabilizeMatching{} can be solved in polynomial time by giving a reduction to \ExactExistsAddMen{}. The reduction is the one-to-one canonical reduction. 
In the constructed instance, all  women (corresponding to workers) will have the same preferences over the men which are copies of a fixed firm. Further,  the target matching for \ExactExistsAddMen{} will be analogous to the given matching $\mu^*$. To prevent any erroneous blocking pairs, for a fixed firm with multiple workers matched to it under $\mu^*$, its most preferred copy will be matched to the worker it prefers most under $\mu^*$ and so on. 

\BudgetAddStabilize*
\begin{proof}

 We shall prove this via a reduction to \ExactExistsAddMen{}.
 Let $\I = \langle F, W, C, \> , \ell, \{\ell_f : f\in F\}, \mu^{*} \rangle$ be the given instance of \RestrictedAddCapacityStabilizeMatching{}, which can be solved by \Cref{alg:EEAM}. We construct an instance  of \ExactExistsAddMen{} as follows:

         \paragraph{Construction.} Let $\I'$ be instance constructed  by the canonical one-to-one reduction of $\I$. We define $P_{add}$ to have men corresponding to the copies of firms from indices $c_f +1$ to $\min (c_f+\ell_f,|\mu^*(f)|)$. Formally,  $P_{orig} = \{p^{f}_i : $ for all $ f \in F, 0 \leq i \leq c_f \}$ and $P_{add} = \{p^{f}_i : $ for all $ f \in F, c_f < i \leq |\mu^*(f)| \}$. Note that if $|\mu^*(f)| \leq c_f$ for some firm $f$, then no men corresponding to that firm will be in $P_{add}$. The preferences over the corresponding agents are same as those in the original instance. Finally, we set the budget for the maximum number of men that can be added as $\ell$.


        We can now define the one-to-one analog $\mu'$ of the given matching $\mu^*$. Let $\mu^*(f)=\{w_1^f,\cdots, w_{|\mu^*(f)|}^f\}$ s.t. $w_1^f\succ_f w_2^f\succ_f \cdots \succ_f w_{|\mu^*(f)|^f}$. 
        Under $\mu'$, for each $i\in [\min (|\mu^*(f)|,c_f+\ell_f)]$ we match the copy $p^f_i$ of the firm to $w_i^f$. That is, the $i^{th}$ copy of $f$ is matched to $f$'s $i^{th}$ most preferred worker under $\mu^*$. 
        All other copies of $f$, namely $p^f_{|\mu^*(f)|},\cdots,p^f_{c_f}$, remain unmatched. This is done because a lower index copy of a firm is more preferred over a higher index copy of the same firm in every woman's preference list. Without this assumption, the matching $\mu'$ cannot be stable.

        \paragraph{Correctness.} We can now prove the correctness of this reduction.

        First, assume that the given instance $\I$ has a solution. That is, there exists a capacity vector $\overline{C}$ satisfying the global and individual budgets such that $\langle F, W, \overline{C}, \> \rangle$ admits a stable matching $\mu$ which is a subset of $\mu^*$. Then we can see that adding men corresponding to $\overline{C}$ in $\I'$ will create the canonical one-to-one instance of $\langle F, W, \overline{C}, \> \rangle$. This instance admits the one-to-one analog of $\mu$, which solves $\I'$. 

        Now assume that the constructed instance $\I'$ admits a solution.  Let $\overline{P} \subseteq P_{add}$ be such that $\langle P_{orig} \cup \overline{P}, Q, \> \rangle$ admits a stable matching $\mu$ which is a subset of $\mu'$. Then we can see that creating a capacity vector $\overline{C}$ corresponding to $P_{orig} \cup \overline{P}$ will satisfy the budget requirements and also solve $\I$. 

        As a result, \RestrictedAddCapacityStabilizeMatching{} reduces to \ExactExistsAddMen{}, and can be solved in polynomial time. 
\end{proof}


We now turn to the unbudgeted version of the problem. Observe that this is simply a  case of the budgeted version where the individual budgets are equal to the global budget.

\problembox{\UnrestrictedAddCapacityStabilizeMatching{}}{
An instance $\I = \langle F,W, C, \> \rangle$,  a global budget $\ell \in \whole$, a matching $\mu^*$ with $|\mu^*(f)| \leq c_f + \ell$ for each $f \in F$. Preferences $\>$ are complete.
}{
Does there exist a capacity vector $\overline{C} \in (\whole)^n$ such that $\overline{C} \geq C, |\overline{C} - C|_1 \leq \ell$ and there is a matching $\mu$ stable for $\I' = \langle W, F, \overline{C}, \> \rangle$ which is a subset of $\mu^*$, i.e., $\mu(f) \subseteq \mu^*(f)$ for each $f$ and $\mu(w) = \mu^*(w)$ or $\emptyset$ for each $w$ ?
}

\begin{theorem}
\UnrestrictedAddCapacityStabilizeMatching{} can be solved in polynomial time.
\label{theorem:UnrestrictedAddCapacityStabilize-Polynomial}
\end{theorem}
\begin{proof}{}

    Given an \UnrestrictedAddCapacityStabilizeMatching{} instance, we can construct an an equivalent \RestrictedAddCapacityStabilizeMatching{} instance with individual budgets equal to the global budget. From this instance, we can construct an equivalent instance of \ExactExistsAddMen{}. We know from \Cref{lemma:ExactExistsAddMen-Polynomial}, this can be solved in polynomial time. It is easy to see that a solution for the constructed instance of \ExactExistsAddMen{} would exactly solve the given instance. As a result, \UnrestrictedDeleteCapacityStabilizeMatching{} can be solved in polynomial time.
\end{proof}

\section{Omitted Material from Section~\ref{sec:Capacity_vs_Preference}}\label{app:capvspref} 

In this section, we will present the missing examples and proofs from \Cref{sec:Capacity_vs_Preference}. We begin this section with the omitted proof for showing that the set of proposals received by a firm under \WPDA{} when its capacity is greater than peak is the same as when its capacity is equal to its peak.

\ProposalsSubset*

\begin{proof}

    Let $w_1,w_2,\cdots, w_t$ be the set of workers who propose to $f$ under \WPDA{} on $\I^s$, where $w_1$ is the first worker to propose to $f$, then $w_2$ and so on. We will show by induction that $w_i$ proposes to $f$ under the \WPDA{} on $\I^r$ for all $i\in [t]$. Specifically, we need to show that each $w_i$ must have been rejected by all firms that she prefers to $f$.

    Firstly, observe that for $i\leq r\leq s$, the \WPDA{} on $a$ and $b$ proceeds identically as the preferences of all agents and the capacities of all firms other than $f$ are identical in $\I^r$ and $\I^s$. Thus, if $t\leq r$, the proof follows immediately. Consequently, we  now assume that $t>r$.
    
    Now consider the base case of $w_{r+1}$. At the point when $w_r$ has proposed to $f$, $f$ is now saturated. However, till this point, everything has proceeded identically in the \WPDA{} on $\I^r$ and $\I^s$. As a result, $w_{r+1}$ continues to be rejected by all firms that it prefers to $f$ and thus, it must propose to $f$ in the \WPDA{} on $\I^r$.

    Assume that for some $i>r+1$, workers $w_1,\cdots, w_{i+1}$ have proposed to $f$ under the \WPDA{} on $\I^r$. Note that the point of difference here with he \WPDA{} on $\I^s$ is that at least one of these workers would now have been rejected under the \WPDA{} on $\I^r$. 
    
    These rejected workers will now propose to other firms acceptable to them, if any. As a result, for worker $w_i$, each firm that she prefers to $f$ will now have either the same or larger set of workers proposing to it. As $w_i$ was rejected under $\I^s$ by each such firm, it must be rejected under $\I^r$. Consequently, $w_i$ must propose to $f$ under the \WPDA{} on $\I^r$.



    
    
\end{proof}

We now present the missing results on the comparisons of \texttt{Delete} and \texttt{Pref} under responsive preferences.  In our comparisons, we first present examples from the lexicographic domain, which clearly extend to the general responsive preferences domain. We will subsequently present the key differences for strongly monotone preferences in \Cref{subsec:Appendix_Strongly_Monotone_Trends}.

\subsection{\texttt{\textup{Delete}} vs \texttt{\textup{Pref}} under \WPDA{}}
\label{subsubsec:Appendix_DeletevsPref_WPDA}

Recall that we compared \texttt{\textup{Delete}} and \texttt{\textup{Pref}} under the \WPDA{} algorithm in \Cref{sec:Capacity_vs_Preference}.  In that discussion, we mentioned~\Cref{thm:prefuseless} which says that \texttt{\textup{Pref}} is ineffective under any stable matching algorithm in the above-peak setting. We will now prove this result.

\prefuseless*

\begin{proof} 
    Let $\I = \langle F,W,C, \> \rangle$ denote the true instance, and $\I' = \langle F,W,C, (\>',\>_{-f}) \rangle$ be any instance where firm $f$ reports a manipulated preference list $\>'_f$. Since $\I$ is an above-peak instance, the firm $f$ does not reject any proposal under the \WPDA{} algorithm on $\I$~(\Cref{prop:Peak_Max_Proposals_WPDA}). The algorithm executes in the same way when firm $f$ permutes the acceptable workers in $\>'_f$. Thus, the set of workers matched with $f$ in the worker-optimal stable matching in $\I'$ is the same as that under $\I$. Thus, in particular, the firm is unsaturated in the worker-optimal stable matching in $\I'$. By the rural hospitals theorem~(\Cref{prop:RuralHospitals}), the firm is matched with the same set of workers in every stable matching in $\I'$. Therefore, regardless of the stable matching algorithm, the firm cannot change the set of workers it is matched with, which proves the claim.
\end{proof}

\Cref{thm:prefuseless} implies that a firm cannot improve via preference manipulation under the \WPDA{} algorithm when its capacity is greater than its peak. We have seen in \Cref{prop:Peak_Max_Proposals_WPDA} that, at peak, the firm does not reject any proposal under the \WPDA{} algorithm. Thus, by a similar argument as in the proof of \Cref{thm:prefuseless}, the firm cannot benefit by manipulating its preferences. We therefore get the following corollary:

\begin{corollary}\label{cor:prefatpeak}
    When the capacity of a firm is greater than or equal to its peak, it cannot improve via preference manipulation under the \WPDA{} algorithm.
\end{corollary}

\Cref{cor:prefatpeak} recovers a result of \citet{ASW15susceptibility}, who showed that a firm can improve via preference manipulation under the \WPDA{} algorithm only if it receives more proposals than its capacity. We have already shown in \Cref{ex:wosm-del-below,ex:wosm-pref-below} that either \texttt{Delete} or \texttt{Pref} could outperform the other below peak. At peak, \cref{example:LP-firm-worse} shows that \texttt{Delete} outperforms the other two manipulations as well. Consequently, we have that \texttt{Delete} can outperform \texttt{Peak} for any capacity of a firm under \WPDA{}. However, \texttt{Pref} can only outperform \texttt{Delete} below peak.



\subsection{\texttt{\textup{Delete}} vs \texttt{\textup{Pref}} under \FPDA{}}
\label{subsubsec:Appendix_DeletevsPref_FPDA}

Let us now compare \texttt{\textup{Delete}} and \texttt{\textup{Pref}} under the \FPDA{} algorithm.

\subsubsection*{Below Peak}
We first show an instance where \texttt{\textup{Pref}} is the best manipulation available to a firm with capacity below the peak. In this instance, \texttt{\textup{Pref}} outperforms both \texttt{\textup{Delete}} and \texttt{\textup{Add}}, and furthermore, \texttt{\textup{Delete}} outperforms \texttt{\textup{Add}} under the \FPDA{} algorithm.

\begin{example}[\texttt{\textup{Pref}} outperforms \texttt{\textup{Delete}} and \texttt{\textup{Add}} under \FPDA{} below peak]\label{ex:fosmpref}
    Consider an instance $\I$ with two firms $f_1, f_2$ and six workers $w_1, \dots, w_6$. The firms have capacities $c_1=2$ and $c_2=3$, and have lexicographic preferences as follows:
    \begin{align*}
        w_1, \,  w_2, \, w_3    &: f_2 \succ f_1 \succ \emptyset & f_1: w_1 \succ w_4 \succ w_5 \succ w_6 \succ w_2 \succ w_3 \\
        w_4, \,  w_5, \, w_6    &: f_1 \succ f_2 \succ \emptyset & f_2: w_4 \succ w_5 \succ w_6 \succ w_1 \succ w_2 \succ w_3
    \end{align*}

The unique stable matching in this instance is
\[\mu_1 =\{ ( \{w_4,w_5\},f_1), (\{w_1,w_2,w_6\}, f_2) \}.\]

Increasing the capacity of $f_1$ to $c_1=3$ would result in the instance $\I'$ with the following unique stable matching:
\[\mu_2 =\{ (\{w_4,w_5,w_6\},f_1), (\{w_1,w_2,w_3\},f_2)\}.\]
It is easy to check that the peak for firm $f_1$ is $3$. Thus, in the original instance $\I$, the firm $f_1$ is below its peak.

Now, starting from the original instance $\I$, let $f_1$ use \texttt{\textup{Delete}} by decreasing its capacity to $c_1 = 1$. As a result, $f_1$ gets matched with the singleton set $\{w_1\}$ under \FPDA{}. On the other hand, if $f_1$ uses \texttt{\textup{Add}} in the instance $\I$ by changing its capacity to $c_1 \geq 3$, then it gets matched with the set $\{w_4,w_5,w_6\}$. Under lexicographic preferences, firm $f_1$ prefers the outcome under \texttt{\textup{Delete}} than that under \texttt{\textup{Add}}.

By using \texttt{\textup{Pref}}, the firm can do even better. Indeed, if $f_1$ misreports its preferences in the instance $\I$ to be
$$f_1 : w_1 \succ w_2 \succ w_4 \succ w_5 \succ w_6 \succ w_3,$$ then under the \FPDA{} algorithm, it is matched with the set $\{w_1,w_2\}$, which is preferable for $f_1$ (according to its true preferences) to the outcome under \texttt{\textup{Delete}}.\qed 
\end{example}

We now show an example where \texttt{\textup{Delete}} can outperform \texttt{\textup{Pref}} under the \FPDA{} algorithm.

\begin{example}[\texttt{\textup{Delete}} outperforms \texttt{\textup{Pref}} under \FPDA{} below peak]
    Consider the following example with three firms $f_1,f_2,f_3$ and four workers $w_1,w_2,w_3,w_4$. The firms have capacities $c_1=c_3=2$ and $c_2=1$, and have lexicographic preferences as follows:
    \begin{align*}
        w_1 &: f_3 \succ f_2 \succ f_1 \succ \emptyset & f_1: w_1 \succ w_2 \succ w_3 \succ w_4 \\
        w_2, \,  w_3, \, w_4 &: f_1 \succ f_2 \succ f_3 \succ \emptyset & f_2: w_2 \succ w_1 \succ w_3 \succ w_4 \\
         && f_3: w_3 \succ w_4 \succ w_1 \succ w_2
    \end{align*}
    The firm-optimal stable matching here is
    \[\mu_1= \{(\{w_2,w_3\},f_1),(w_4,f_2),(w_1,f_3)\}.\]

    It is easy to verify that $f_1$ cannot get matched with $w_1$ no matter which permutation of its preferences it reports. Indeed, $f_1$ must make two proposals, and one of these proposals must be to $w_2$, $w_3$, or $w_4$. The latter proposal will displace $f_2$ or $f_3$, causing that firm to propose to $w_1$ and consequently displace $f_1$. In other words, \texttt{\textup{Pref}} is unhelpful.
    
    
    On the other hand, if $f_1$ decreases its capacity to $1$, the \FPDA{} outcome will be:
    \[  \mu_2 =\{(w_1,f_1), (w_2,f_2),(w_3,f_3),(w_4,f_4)\}.\]

    Due to lexicographic preferences, $f_1$ prefers $\mu_2$ to $\mu_1$, implying that \texttt{\textup{Delete}} outperforms \texttt{\textup{Pref}}.
    \qed
\end{example}

\begin{example}[\texttt{\textup{Delete}} outperforms \texttt{\textup{Add}} under \FPDA{} below peak]
    Consider the following example with three firms $f_1,f_2,f_3$ and six workers $w_1,w_2,w_3,w_4,w_5,w_6$. The firms have capacities $c_1=c_3=2$ and $c_2=1$, and have lexicographic preferences as follows:
    \begin{align*}
        w_1 &:  f_2 \succ f_1 \succ \emptyset                   && f_1: w_1 \succ w_2 \succ w_3 \succ w_4 \succ w_5 \succ w_6 \\
        w_2, \,  w_3 &: f_3 \succ f_1 \succ f_2 \succ \emptyset && f_2: w_2 \succ w_3 \succ w_1 \succ \emptyset \\
        w_4, \,  w_5, \, w_6&: f_1 \succ f_3 \succ \emptyset    && f_3: w_4 \succ w_5 \succ w_6 \succ w_2 \succ w_3
    \end{align*}
    The firm-optimal stable matching here is
    \[\mu_1= \{(\{w_2,w_3\},f_1),(w_1,f_2),(\{w_4,w_5,w_6\},f_3)\}.\]

    If $f_1$ had capacity $3$, the unique stable matching would be 

    \[\mu_2=\{(\{w_4,w_5,w_6\},f_1),(w_1,f_2),(\{w_2,w_3\},f_3)\}\]
    
    Thus, when $c_{f_1}=2$, $f_1$ is below peak and due.
    Further, if $f_1$ decreases its capacity to $1$, the \FPDA{} outcome will be:
    \[  \mu_3 =\{(w_1,f_1), (w_2,f_2),(\{w_4,w_5,w_6\},f_4)\}.\]

    Due to lexicographic preferences, $f_1$ prefers $\mu_3$ to $\mu_1$ to $\mu_2$, implying that \texttt{\textup{Delete}} outperforms \texttt{\textup{Add}} in this instance.
    \qed
\end{example}

\subsubsection*{At Peak}
Unlike \WPDA{}, a firm \emph{can} profitably manipulate by misreporting its preferences in the at-peak regime under the \FPDA{} algorithm. 
We now show a straightforward example for lexicographic preferences where \texttt{\textup{Pref}} outperforms \texttt{\textup{Delete}}, which, in turn, outperforms \texttt{\textup{Add}} under the \FPDA{} algorithm. This example can be easily extended to the strongly monotone setting.

\begin{example}[\texttt{\textup{Pref}} outperforms \texttt{\textup{Delete}} and \texttt{\textup{Add}} under \FPDA{} at peak]
\label{ex:prefatpeak}
    Consider the following instance with two firms $f_1$, $f_2$ with lexicographic preferences and four workers $w_1,w_2,w_3,w_4$. The preference relations are as follows:
    \begin{align*}
        w_1 ,\,  w_4 &: f_2 \succ f_1 \succ \emptyset & f_1: w_1 \succ w_2 \succ w_3 \succ w_4 \\
        w_2, \,  w_3 &: f_1 \succ f_2 \succ \emptyset & f_2: w_2 \succ w_3 \succ w_1 \succ w_4
    \end{align*}
    The capacities of the firms are $c_1=c_2=2$. For this instance, $f_1$ is at its peak capacity. The unique stable matching for this instance is
    \[\mu_1=\{(\{w_2,w_3\},f_1),(\{w_1,w_4\},f_2)\}.\]
    Increasing the capacity of $f_1$ will also result in the same stable matching $\mu_1$. Thus, \texttt{\textup{Add}} is unhelpful. However, if $f_1$ misreports its preference as 
    $$f_1: w_1 \succ w_4 \succ w_2 \succ w_3,$$ the resulting firm-optimal stable matching is
    \[\mu_2=\{(\{w_1,w_4\},f_1),(\{w_2,w_3\},f_2)\},\]
    implying that \texttt{\textup{Pref}} outperforms \texttt{\textup{Add}}.

    It can be similarly observed that \texttt{\textup{Pref}} outperforms \texttt{\textup{Delete}}. Indeed, if $f_1$ decreases its capacity to $c_1 = 1$ in the original instance, then it gets matched with $\{w_1\}$ in the firm-optimal matching. While this is an improvement over its true outcome $\{w_2,w_3\}$ (and thus \texttt{\textup{Delete}} outperforms \texttt{\textup{Add}}), it is still worse than the outcome $\{w_1,w_4\}$ achieved under \texttt{\textup{Pref}}.
    \qed
\end{example}

We can also show that \texttt{\textup{Delete}} outperforms \texttt{\textup{Pref}} under the \FPDA{} algorithm in the at-peak setting. Intuitively, \texttt{\textup{Pref}} fails to be helpful because the firm is forced to make suboptimal proposals, which are avoided under \texttt{\textup{Delete}}. Indeed, consider the instance $\I''$ in \Cref{example:LP-firm-worse} where both firms have capacity $c_1 = c_2 = 2$. Firm $f_1$ gets matched with $\{w_2,w_3\}$ in the firm-optimal stable matching in $\I''$. By reducing its capacity to $c_1 = 1$, the firm-optimal outcome for $f_1$ changes to $\{w_1\}$, which is strictly better than $\{w_2,w_3\}$ due to lexicographic preferences. However, \texttt{\textup{Pref}} is unhelpful for $f_1$ in $\I''$. This is because $f_1$ is forced to make two proposals under the \FPDA{} algorithm. One of these proposals will cause $f_2$ to be rejected by $w_2$ or $w_3$. This results in firm $f_2$ proposing to $w_1$, who then rejects $f_1$ in favor of $f_2$. Therefore, firm $f_1$ cannot be matched with $w_1$ under \texttt{\textup{Pref}}.


\subsubsection*{Above Peak} As shown in \Cref{thm:prefuseless}, \texttt{\textup{Pref}} is unhelpful under the \FPDA{} algorithm when the capacity of the firm is above its peak. However, \texttt{\textup{Delete}} can outperform \texttt{\textup{Pref}} under \FPDA{} in the above-peak setting. We can show this through \Cref{example:LP-firm-worse}. Consider the instance $\I''$ where both firms have capacity $c_1 = c_2 = 2$. The firm-optimal outcome for firm $f_2$ is $\{w_1\}$. By decreasing its capacity to $c_2 = 1$, firm $f_2$ can improve its firm-optimal outcome to $\{w_3\}$. However, the firm cannot improve its outcome through \texttt{\textup{Pref}} because of the following reason: In the \FPDA{} algorithm, $f_2$ is required to make two proposals. At least one of these proposals must be made to $w_1$ or $w_2$. If $f_2$ proposes to $w_2$, it is rejected in favor of $f_1$. Thus, $f_2$ must propose to $w_1$, who prefers $f_2$ over $f_1$. Thus, $w_1$ rejects $f_1$, which causes $f_1$ to propose to $w_3$, preventing $f_2$ from being matched with $w_3$. The same example also shows that \texttt{\textup{Delete}} outperforms \texttt{\textup{Add}} because, by increasing its capacity, firm $f_2$ cannot change its \FPDA{} outcome.

\subsection{Strongly Monotone Preferences}
\label{subsec:Appendix_Strongly_Monotone_Trends}

We now shift focus to a subclass of responsive preferences called \emph{strongly monotone} preferences~\citep{KU06games}. In this class, agents always prefer a set over its strict subsets. Preferences among sets of equal cardinality can be arbitrary as long as they are responsive. 

\citet{KU06games} and \citet{kojima2007can} have shown that under the \WPDA{} algorithm, a firm cannot benefit by \texttt{\textup{Delete}} when its preferences are strongly monotone. 
\begin{proposition}[\citealp{KU06games,kojima2007can}]
    Under strongly monotone preferences, a firm cannot improve via deleting capacity under the \WPDA{} algorithm.
    \label{prop:Delete-Useless-WPDA-Strongly-Monotone}
\end{proposition}
This result is in sharp contrast with our results for the general responsive preference class~(\Cref{fig:comparison}). This observation further motivates the study of strongly monotone (responsive) preferences to understand which manipulation actions can be useful. \Cref{fig:smpcomparison} presents the manipulation trends for strongly monotone preferences.

\begin{figure}[t]
\centering
    \tikzset{every picture/.style={line width=1pt}}  

    \begin{subfigure}[b]{0.1\linewidth}
	\centering
	\begin{tikzpicture}
		\node[] (1) at (0,2) {{\WPDA{}:}};
            \node[] (1) at (0,0) { };
	\end{tikzpicture}
    \end{subfigure}
    \begin{subfigure}[b]{0.28\linewidth}
	\centering
	\begin{tikzpicture}
		\tikzset{mynode/.style = {shape=circle,draw,inner sep=0pt,minimum size=25pt}} 
		\node[mynode] (1) at (1.5,1.5) {\footnotesize{Pref}};
		\node[mynode] (2) at (2.6,0) {\footnotesize{Del}};
		\node[mynode] (3) at (0.4,0) {\footnotesize{Add}};
		    \draw[->,-latex] (1) -- (2);  
            \draw[->,-latex] (3) -- (2);
            \draw[->,-latex] (3) -- (1); 
	\end{tikzpicture}
	\caption{Below peak}
    \end{subfigure}
    \begin{subfigure}[b]{0.28\linewidth}
	\centering
	\begin{tikzpicture}
		\tikzset{mynode/.style = {shape=circle,draw,inner sep=0pt,minimum size=25pt}} 
		\node[mynode] (1) at (1.5,1.5) {\footnotesize{Pref}};
		\node[mynode] (2) at (2.6,0) {\footnotesize{Del}};
		\node[mynode] (3) at (0.4,0) {\footnotesize{Add}};
	\end{tikzpicture}
	\caption{At peak}
    \end{subfigure}
    \begin{subfigure}[b]{0.28\linewidth}
	\centering
	\begin{tikzpicture}
		\tikzset{mynode/.style = {shape=circle,draw,inner sep=0pt,minimum size=25pt}} 
		\node[mynode] (1) at (1.5,1.5) {\footnotesize{Pref}};
		\node[mynode] (2) at (2.6,0) {\footnotesize{Del}};
		\node[mynode] (3) at (0.4,0) {\footnotesize{Add}};
	\end{tikzpicture}
	\caption{Above peak}
    \end{subfigure}
    \begin{subfigure}[b]{\linewidth}
	\centering
	\begin{tikzpicture}
		\node (1) at (0,0) {};
	\end{tikzpicture}
    \end{subfigure}
    \begin{subfigure}[b]{0.1\linewidth}
	\centering
	\begin{tikzpicture}
		\node[] (1) at (0,2) {{\FPDA{}:}};
            \node[] (1) at (0,0) { };
	\end{tikzpicture}
    \end{subfigure}
    \begin{subfigure}[b]{0.28\linewidth}
	\centering
	\begin{tikzpicture}
		\tikzset{mynode/.style = {shape=circle,draw,inner sep=0pt,minimum size=25pt}} 
		\node[mynode] (1) at (1.5,1.5) {\footnotesize{Pref}};
		\node[mynode] (2) at (2.6,0) {\footnotesize{Del}};
		\node[mynode] (3) at (0.4,0) {\footnotesize{Add}};
		    \draw[->,-latex] (1) -- (2);  
            \draw[->,-latex] (3) -- (2);
            \draw[->,-latex] (3) -- (1);
	\end{tikzpicture}
	\caption{Below peak}
    \end{subfigure}
    \begin{subfigure}[b]{0.28\linewidth}
	\centering
	\begin{tikzpicture}
		\tikzset{mynode/.style = {shape=circle,draw,inner sep=0pt,minimum size=25pt}} 
		\node[mynode] (1) at (1.5,1.5) {\footnotesize{Pref}};
		\node[mynode] (2) at (2.6,0) {\footnotesize{Del}};
		\node[mynode] (3) at (0.4,0) {\footnotesize{Add}};
            \draw[->,-latex] (1) -- (2);
            \draw[->,-latex] (1) -- (3); 
	\end{tikzpicture}
	\caption{At peak}
    \end{subfigure}
    \begin{subfigure}[b]{0.28\linewidth}
	\centering
	\begin{tikzpicture}
		\tikzset{mynode/.style = {shape=circle,draw,inner sep=0pt,minimum size=25pt}} 
		\node[mynode] (1) at (1.5,1.5) {\footnotesize{Pref}};
		\node[mynode] (2) at (2.6,0) {\footnotesize{Del}};
		\node[mynode] (3) at (0.4,0) {\footnotesize{Add}};
            \draw[->,-latex] (2) -- (1);
            \draw[->,-latex] (2) -- (3);
	\end{tikzpicture}
	\caption{Above peak}
    \end{subfigure}
    \vspace{-2mm}
    
\caption{\small{Manipulation trends for \emph{strongly monotone preferences} for the \WPDA{} (top) and \FPDA{} (bottom) algorithm under the below peak/at peak/above peak regimes. An arrow from action $X$ to action $Y$ denotes the existence of an instance where $X$ is strictly more beneficial for the firm than $Y$. Each missing arrow from $X$ to $Y$ denotes that there is (provably) no instance where $X$ is more beneficial than $Y$.}}
\label{fig:smpcomparison}
\end{figure}

\subsubsection*{Manipulation via \texttt{\textup{Add}}} 
Let us first consider the case where the firm's capacity is below peak. We have already shown in \Cref{ex:masterlist} that under strongly monotone preferences, \texttt{Add} can outperform \texttt{Delete} and \texttt{Pref} below peak for all stable matching algorithms. We now prove a stronger result, that \texttt{Add} always outperforms \texttt{Delete} and \texttt{Pref} below peak for this setting.  We find that when \Cref{lem:convergebeyondpeak} is combined with the fact that under strongly monotone preferences, larger sets are always preferred, we get the following result. 

\begin{proposition}\label{prop:SMPadd}
    Given an instance $\I=\langle F,W,C, \succ \rangle$ and a firm $f\in F$, with strongly monotone preferences and  $c_f< p_f$. The best manipulation strategy for  $f$ is to increase its capacity to $p_f$.
\end{proposition}

\begin{proof}
    Firstly, we note that the best reported capacity for $f$ is $p_f$. That is, for any stable matching algorithm, $f$ will (weakly) prefer the matching returned by the algorithm when it misreports its capacity as $p_f$ over the matching returned on any other reported capacity by $f$. Thus, whenever $c_f<p_f$, \texttt{Add} will always outperform \texttt{Delete}.
    
    Further, observe that under any stable matching algorithm,  \texttt{Pref} will always have $f$ matched to a set of workers of size $c_f<p_f$. In contrast, if $f$ were to report its capacity as $p_f>c_f$, it would be matched to a strictly larger set.  As a result, \texttt{Add} will always outperform \texttt{Pref} below peak.
    
\end{proof}

As a result, below peak, the only useful manipulation is \texttt{\textup{Add}}. When the firm's capacity is at or  above peak, we know from \Cref{cor:Add-Useless-At-or-Above-Peak} that \texttt{\textup{Add}} is not useful for {\em any} stable matching instance, even strongly monotone ones.

\subsubsection*{Manipulation via \texttt{\textup{Delete}}}

In contrast to the unrestricted case, we find that for strongly monotone preferences, \texttt{\textup{Delete}} is not useful below or at peak. 

\begin{proposition}\label{prop:delsmp}
    Given an instance $\I=\langle F,W,C, \succ \rangle$ and a firm $f\in F$. If $f$ has strongly monotone preferences when $c_f\leq p_f$, $f$ cannot benefit from reducing its capacity under any stable matching algorithm. 
\end{proposition}

\begin{proof}
    Recall that for a firm $f$ with a strongly monotone preference relation $\succ_f$, for any subsets of workers acceptable to $f$ $W_1$ and $W_2$, whenever $|W_1|>|W_2|$ then $W_1\succ_f W_2$. 

    Fix an instance $\I$ and firm $f$ such that $c_f\leq p_f$. Let $M$ be any stable matching algorithm. Recall that $\I^b_f$ is the instance identical to $\I$ with the capacity of $f$ set to $b$. 
    For any $b<c_f<p_f$, the size of the set of workers matched to $f$ under $M(I_{f}^b)$ will be $b$ which is less than the size of the set of workers matched to $f$ under $M(I)$. As $f$'s preference is strongly monotone, $f$ prefers the larger set and cannot benefit by reducing its capacity.

    Now consider an arbitrary instance $\I$ and firm $f$ with a strongly monotone preference. We know that if $c_f\leq p_f$, $f$ cannot benefit by reducing its capacity for any stable matching algorithm, in particular with the  \WPDA{} algorithm. If $c_f>p_f$, from Lemma \ref{lem:convergebeyondpeak}, we have that the set of agents matched to $f$ under the worker-optimal stable matching of $\I$ is the same as that under the worker-optimal stable matching of $\I_{f}^{p_f}$. 
    
    Thus, for the \WPDA{} algorithm, if $f$ were to reduce its capacity to $b$ s.t. $c_f>b\geq p_f$, it would continue to be matched to the same set of agents. If the choice of reduced capacity were $b<p_f$, then the size of the set of agents matched to $f$ would be smaller, which it does not prefer. 
    
    Consequently, a firm with strongly monotone preferences can never benefit by reducing its capacity when the stable matching algorithm is the \WPDA{} algorithm.
\end{proof}

The remaining case for \texttt{\textup{Delete}} is when the firm's capacity is above peak. Here, we find that we get different results for \WPDA{} and \FPDA{}. Firstly recall that from \Cref{lem:convergebeyondpeak}, we have that for a firm with capacity above peak, every stable matching matches it to the same set of workers as the \WPDA{} when its capacity is set to peak. As a result, when the firm reports a lower capacity, it will be matched to the same or fewer number of workers. As a result, we get the following observation.

\begin{observation}\label{obs:SMP-del-WPDA}
    For a firm $f$ with strongly monotone preferences $c_f>p_f$, under \WPDA{} the \texttt{\textup{Delete}} will always have $f$ matched to a weakly worse set of workers.
\end{observation}

A consequence of \Cref{prop:delsmp}, \Cref{obs:SMP-del-WPDA} and the fact that \texttt{\textup{Pref}} is not useful under the \WPDA{} algorithm at and above peak (shown in \Cref{thm:prefuseless}) is the following.

\begin{remark}
    For a firm with strongly monotone preference and capacity at or above the peak, there is no beneficial manipulation via preferences or capacities under the \WPDA{} algorithm.
\end{remark}

For the \FPDA{} however, we find that \texttt{\textup{Delete}} can be useful above peak. In fact, we now show that under strongly monotone preferences, \texttt{\textup{Delete}} is useful for a firm $f$ above peak if and only if the worker and firm optimal matchings for the instance  $\I^{p_f}_f$ match $f$ to distinct sets of workers.

\begin{proposition}\label{prop:smp-del-above-peak}
    Given an instance $\I=\langle F,W,C, \succ \rangle$ and a firm $f\in F$ such that $c_f>p_f$ and $f$ has strongly monotone preferences. Under the \FPDA{}, $f$ can benefit by reducing its capacity if and only if $f$ is such that $\mu^W(f)\neq \mu^F(f)$ where $\mu^W$ and $\mu^F$ are the worker and firm optimal stable matchings, respectively under $\I^{p_f}_f$.  
\end{proposition}

\begin{proof}
    Firstly, consider the case where $f$ is such that $\mu^W(f)\neq \mu^F(f)$. From \Cref{lem:convergebeyondpeak}, we know that under the \FPDA{} on the true instance $\I$, $f$ is matched to $\mu^W(f)$.Further, by setting its capacity to $p_f<c_f$, $f$ would be matched to $\mu^F(f)$. By definition of firm and worker optimal stable matchings, $\mu^F(f)\succ_f\mu^W(f)$. Thus, $f$ will be matched to a strictly better set by reducing its capacity to $p_f$.

    Conversely, let $f$ be such that $\mu^W(f)$ is the same as $\mu^F(f)$. Here, for any $b$ s.t. $p_f\leq b<c_f$, under the \FPDA{}, and indeed any stable matching algorithm, we have that on setting $f$'s capacity to $b$, $f$ would be matched to $\mu^W(f)$. Here the size of this set is $|\mu^W(f)|=p_f$. Now for any $b<p_f$, $f$ would only be matched to a set of size $b$. As $f$ has strongly monotone preferences, $f$ will strictly prefer $\mu^W(f)$ to this. As a result, when $\mu^W(f)= \mu^F(f)$, we have that \texttt{\textup{Delete}} can only match $f$ to a weakly worse set than outcome of the true capacity under the \FPDA{}.
\end{proof}

Observe that the proof of \Cref{prop:smp-del-above-peak} shows that whenever  $\mu^W(f)$ is the same as $\mu^F(f)$, under {\em any} stable matching algorithm, \texttt{\textup{Delete}} is not useful. It is only possible for it be useful if by reducing its capacity to $p_f$, $f$ can be matched to a set of workers strictly better than $\mu^W(f)$.  We give one such instance in the following example.

\begin{example}[\texttt{\textup{Delete}} outperforms \texttt{\textup{Pref}} and \texttt{\textup{Add}} under \FPDA{} above peak]
\label{ex:SMPdelabovepeak}
    Consider the following instance with two firms $f_1$, $f_2$ with strongly monotone preferences and three workers $w_1,w_2,w_3$. The preference relations are as follows:
    \begin{align*}
        w_1 , \, w_2    &: f_2 \succ f_1 \succ \emptyset && f_1: \{w_1,w_2,w_3\}\succ \{w_1,w_2\}\succ \{w_1,w_3\}\succ\{w_2,w_3\}\succ \cdots \\
       w_3              &: f_1 \succ f_2 \succ \emptyset &&  f_2:  \{w_3\}\succ \{w_2\}\succ\{w_1\}.
    \end{align*}
    The capacities of the firms are $c_1=3$ and $c_2=1$. For this instance, $p_{f_1}$ is clearly $2$. The unique stable matching for this instance is
    \[\mu_1=\{(\{w_1,w_3\},f_1),(w_2,f_2)\}.\]
    Decreasing the capacity of $f_1$ to $2$ will result in the following stable matching under \FPDA{}:
    \[\mu_2=\{(\{w_1,w_2\},f_1),(w_3,f_2)\}.\]

    Clearly, $f_1$ strictly prefers $\mu_2$ to $\mu_1$. We know from \Cref{cor:Add-Useless-At-or-Above-Peak} and \Cref{thm:prefuseless} that \texttt{\textup{Add}} and \texttt{\textup{Pref}} are not useful here. 
    \qed
\end{example}

\subsubsection*{Manipulation via \texttt{\textup{Pref}}}

We now study when \texttt{\textup{Pref}} can be useful to agents with strongly monotone preferences. Above peak, we have seen in Theorem \ref{thm:prefuseless} that preference manipulation cannot improve the stable set of a firm. Thus, the remaining cases are of below and at peak.

\paragraph{Below peak.} For strongly monotone preferences, when the capacity of a firm is below the peak, clearly, increasing the capacity of the firm is the best manipulation action, as seen in \Cref{prop:SMPadd}. 
However, \texttt{\textup{Pref}} can still outperform \texttt{\textup{Delete}} below peak. We first give an example for this under \WPDA{}.

\begin{example}[\texttt{\textup{Pref}} outperforms \texttt{\textup{Delete}} under \WPDA{} below peak]\label{ex:SMP-pref-WPDA}
     Consider an instance $\I$ with three firms $f_1, f_2, f_3$ and four workers $w_1,w_2,w_3, w_4$. The firms have unit capacities (i.e., $c_1=c_2=c_3=1$). The workers preferences are:
    \begin{align*}
        w_1  &: f_2 \succ f_1 \succ f_3\succ  \emptyset \\
        w_2 \, , w_3 &: f_1 \succ f_2 \succ f_3 \succ \emptyset\\
        w_4 &: f_3 \succ f_1 \succ f_2 \succ \emptyset
    \end{align*}
    As agents have strongly monotone preferences and unit capacities, we only need to know their preferences over individual workers for this example. These are:
    \begin{align*}
        f_1& : w_4 \succ w_1 \succ w_2 \succ w_3 \\
        f_2& : w_3 \succ w_2 \succ w_1 \succ w_4 \\
        f_3& : w_1 \succ w_4 \succ w_2 \succ w_3
    \end{align*}
    
    Under the \WPDA{} algorithm, firm $f_1$ is matched with $\{w_1\}$. 
    %
    %
    If $f_1$ uses \texttt{\textup{Add}} by switching to any capacity $c_1 \geq 2$, its \WPDA{} match is the set $\{w_2,w_3\}$. It is easy to verify that the peak for firm $f_1$ is $p_f(\I) = 2$. Thus, under $\I$, the capacity of firm $f_1$ is below the peak.
    %
    
    If $f_1$ uses \texttt{\textup{Pref}} in the instance $\I$ by misreporting its preferences to be $w_4 \succ w_2 \succ w_1 \succ w_3$, then its \WPDA{} match is $\{w_4\}$. On the other hand, using \texttt{\textup{Delete}} in the instance $\I$ (by reducing the capacity to $c_1=0$) is the worst outcome for $f_1$ as it is left unmatched.\qed
\end{example}

We now give an example where \texttt{Pref} outperforms \texttt{Delete} below peak under the \FPDA{}.

\begin{example}[\texttt{\textup{Pref}} outperforms \texttt{\textup{Delete}} under \FPDA{} below peak]\label{ex:SMP-pref-FPDA}
    Consider an instance $\I$ with two firms $f_1, f_2$ and six workers $w_1, \dots, w_6$. The firms have capacities $c_1=2$ and $c_2=3$. The workers preferences are as follows:
    \begin{align*}
        w_1, \,  w_2, \, w_3    &: f_2 \succ f_1 \succ \emptyset  \\
        w_4, \,  w_5, \, w_6    &: f_1 \succ f_2 \succ \emptyset 
    \end{align*}

    The firm preferences are strongly monotone with tie breaking among sets of equal size as they would be under lexicographic preferences based on the preferences over individual workers:

    \begin{align*}
        f_1 &: w_1  \succ w_4 \succ w_5 \succ w_6 \succ w_2 \succ w_3 \\
        f_2 & : w_4 \succ w_5 \succ w_6 \succ w_1 \succ w_2 \succ w_3
    \end{align*}

That is, given two sets of workers of different sizes, $f_1$ and $f_2$ both prefer the larger set. Given two sets of same size, $f_1$ prefers the set which has its most preferred worker not contained in the other. Specifically, given a choice of $\{w_4,w_5\}$ and $\{w_1,w_2\}$, $f_1$ prefers $\{w_1,w_2\}$. 
The unique stable matching in this instance is
\[\mu_1 =\{ ( \{w_4,w_5\},f_1), (\{w_1,w_2,w_6\}, f_2) \}.\]

Increasing the capacity of $f_1$ to $c_1=3$ would result in the instance $\I'$ with the following unique stable matching:
\[\mu_2 =\{ (\{w_4,w_5,w_6\},f_1), (\{w_1,w_2,w_3\},f_2)\}.\]
It is easy to check that the peak for firm $f_1$ is $3$. Thus, in the original instance $\I$, the firm $f_1$ is below its peak.

Now, starting from the original instance $\I$, let $f_1$ use \texttt{\textup{Delete}} by decreasing its capacity to $c_1 = 1$. As a result, $f_1$ gets matched with the singleton set $\{w_1\}$ under \FPDA{}. 

By using \texttt{\textup{Pref}}, the firm can do even better. Indeed, if $f_1$ misreports its preferences in the instance $\I$ to be
$$f_1 : w_1 \succ w_2 \succ w_4 \succ w_5 \succ w_6 \succ w_3,$$ then under the \FPDA{} algorithm, it is matched with the set $\{w_1,w_2\}$, which is preferable for $f_1$ (according to its true preferences) to the outcome under \texttt{\textup{Delete}}.\qed 
\end{example}

\paragraph{At peak.} Under the \WPDA{} algorithm, we have seen that preference manipulation in \Cref{cor:prefatpeak} that preference manipulation is not useful at peak. We now show an example where preference manipulation can be useful under strongly monotone preferences at peak under the \FPDA{} algorithm. 

\begin{example}[\texttt{\textup{Pref}} outperforms \texttt{\textup{Add}} under \FPDA{}]\label{ex:prefsmp}
    Consider the following instance with two firms $f_1,f_2$ and four workers $w_1,w_2,w_3,w_4$. The preference relations are as follows:
    \begin{align*}
        w_1 ,\,  w_4 &: f_2 \succ f_1 \succ \emptyset\\
        w_2, \,  w_3 &: f_1 \succ f_2 \succ \emptyset
    \end{align*}
    \begin{align*}
        f_1 &: \{w_1,w_2\} \succ \{w_1, w_3\} \succ \{w_1,w_4\} \succ \{w_2,w_3\} \succ \\
        & \hspace{0.2in} \{w_2,w_4\}\succ \{w_3,w_4\}\succ \{w_1\}\succ \{w_2\} \succ \{w_3\} \succ \{w_4\} \\
        f_2 &: \{w_2,w_3\} \succ \{w_1, w_2\} \succ \{w_2,w_4\} \succ \{w_1,w_3\} \succ \\
        & \hspace{0.2in} \{w_3,w_4\}\succ \{w_1,w_4\}\succ \{w_2\}\succ \{w_3\} \succ \{w_1\} \succ \{w_4\} 
    \end{align*}

    Here, the true capacity of the firms is $c_1=c_2=2$. For this instance, $f_1$ is at its peak capacity. The unique stable matching for this instance is
    \[\mu_1=\{(\{w_2,w_3\},f_1),(\{w_1,w_4\},f_2)\}.\]
    Increasing the capacity of $f_1$ will also result in the same stable matching $\mu_1$ (hence, \texttt{\textup{Add}} is ineffective). However, if $f_1$ were to misreport its preference as $\{w_1,w_4\} \succ \{w_1, w_2\} \succ \{w_1,w_3\} \succ \{w_2,w_4\} \succ \{w_2, w_3\}\succ \{w_2,w_3\}\succ \{w_1\}\succ \{w_4\} \succ \{w_2\} \succ \{w_3\}$, the resulting firm-optimal stable matching will be
    \[\mu_2=\{(\{w_1,w_4\},f_1),(\{w_2,w_3\},f_2)\}.\]
    Clearly, it is better for $f_1$ to misreport its preference.\qed
\end{example}
\clearpage

\end{document}